\documentclass[a4paper, reqno]{amsart}

\usepackage{amsmath,amsfonts,amssymb, amsthm, mathtools}
\usepackage{xcolor}
\definecolor{dark2Green}{HTML}{1B9E77}
\definecolor{dark2Orange}{HTML}{D95F02}
\definecolor{dark2Purple}{HTML}{7570B3}
\definecolor{dark2Pink}{HTML}{E7298A}
\usepackage[colorlinks=true,linkcolor=dark2Pink, citecolor=dark2Green, urlcolor=dark2Purple]{hyperref}
\usepackage{graphicx}
\usepackage{dsfont}
\usepackage{bbm}
\usepackage{units}
\usepackage{tikz, tikz-3dplot}
\usepackage{pgfplots,pgfplotstable}
\usepackage{cleveref}
\usepackage{microtype}

\usepackage{float}
\pgfplotsset{compat=newest}
\usetikzlibrary{shapes.geometric, arrows, arrows.meta, positioning, calc, backgrounds}
\usepackage[export]{adjustbox}

\tikzstyle{midbox} = [rectangle, rounded corners, minimum width=2.62cm, minimum height=.7cm,text centered, draw=black, fill=black!30]
\tikzstyle{arrow} = [thick,->,>=stealth]
\usepackage[breakable,skins,most]{tcolorbox}
\newtcolorbox{supplementbox}{
    breakable,
    enhanced,
    width=\linewidth,
    sharp corners=all,
    colback=white!97!black,
    colframe=black,
    boxrule=0.4pt,
    borderline={0.4pt}{2pt}{black},
}

\usetikzlibrary{external}
\tikzset{external/system call={lualatex
    -shell-escape
    \tikzexternalcheckshellescape 
    -halt-on-error -interaction=batchmode
    -jobname "\image" "\texsource"}
    }
\newcommand{\R}{\ensuremath{\mathbb R}}  
\newcommand{\N}{\ensuremath{\mathbb N}} 

\newcommand{\dd}{\ell}

\newcommand{\one}{\mathds{1}}

\renewcommand{\k}{\mathsf{k}}

\DeclareMathOperator*{\argmin}{arg\,min}

\newtheorem{theorem}{Theorem}[section]
\newtheorem{lemma}[theorem]{Lemma}
\newtheorem{corollary}[theorem]{Corollary}
\newtheorem{remark}[theorem]{Remark}
\newtheorem{proposition}[theorem]{Proposition}

\title[Excitation of control-affine systems and Koopman error bounds]{Excitation of control-affine systems and Koopman error bounds}

\author[Schmitz et al.]{Philipp Schmitz}
\address{Optimization-based Control Group, Institute of Mathematics, Technische Universität Ilmenau, Ilmenau, Germany}
\email{\href{mailto:philipp.schmitz@tu-ilmenau.de}{philipp.schmitz@tu-ilmenau.de}}
\thanks{PS gratefully acknowledges support from the Carl Zeiss Foundation (VerneDCt – Project No. 2021-10-003). %
LB gratefully acknowledges funding by the German Research Foundation (DFG; project numbers~$535860958$ and~$545246093$). %
FMP was funded by the DFG, project number 519323897. He was further supported by the free state of Thuringia and the German Federal Ministry of Education and Research (BMBF) within the project {\it THInKI--Th\"uringer Hochschulinitiative für KI im Studium}. %
PE and MR gratefully acknowledge support from the German Research Foundation (DFG; project numbers~$545246093$ and~$433183605$). %
HE is grateful for the past support from the German Research Foundation (DFG; project number~$545246093$).}
\author[]{Lea Bold}
\address{Optimization-based Control Group, Institute of Mathematics, Technische Universität Ilmenau, Ilmenau, Germany}
\email{\href{mailto:lea.bold@tu-ilmenau.de}{lea.bold@tu-ilmenau.de}}
\author[]{Friedrich M.\ Philipp}
\address{Optimization-based Control Group, Institute of Mathematics, Technische Universität Ilmenau, Ilmenau, Germany}
\email{\href{mailto:friedrich.philipp@tu-ilmenau.de}{friedrich.philipp@tu-ilmenau.de}}
\author[]{Mario Rosenfelder}
\address{Institute of Engineering and Computational Mechanics, University of Stuttgart, Stuttgart, Germany}
\email{\href{mailto:mario.rosenfelder@uni-stuttgart.de}{mario.rosenfelder@uni-stuttgart.de}}
\author[]{Peter Eberhard}
\email{\href{mailto:peter.eberhard@uni-stuttgart.de}{peter.eberhard@uni-stuttgart.de}}
\author[]{Henrik Ebel}
\address{Department of Mechanical Engineering, LUT University, Lappeenranta, Finland}
\email{\href{mailto:henrik.ebel@lut.fi}{henrik.ebel@lut.fi}}
\author[]{Karl Worthmann}
\address{Optimization-based Control Group, Institute of Mathematics, Technische Universität Ilmenau, Ilmenau, Germany}
\email{\href{mailto:karl.worthmann@tu-ilmenau.de}{karl.worthmann@tu-ilmenau.de}}

\begin{document}
\tikzexternaldisable
	
	\begin{abstract}
	The Koopman operator and extended dynamic mode decomposition (EDMD) as a data-driven technique for its approximation have attracted considerable attention as a key tool for modeling, analysis, and control of complex dynamical systems. 
    However, extensions towards control-affine systems resulting in bilinear surrogate models are prone to demanding data requirements rendering their applicability intricate. 
    In this paper, we propose a framework for data-fitting of control-affine mappings to increase the robustness margin in the associated system identification problem and, thus, to provide reliable bilinear EDMD schemes.
    In particular, guidelines for input selection based on subspace angles are deduced such that a desired threshold with respect to the minimal singular value is ensured. Moreover, we derive necessary and sufficient conditions of optimality for maximizing the minimal singular value. 
    Further, we demonstrate the usefulness of the proposed approach using bilinear EDMD with control for nonholonomic robots.
	\end{abstract}
    
    \maketitle
	
	\section{Introduction}
    
As an idea, system identification can arguably be considered as old as mathematical system modeling in general, considering that even many fundamental laws of physics were identified from observation data. 
As such, the identification of systems from measurement data has long been a tried-and-tested procedure for engineers, with approaches being as varied as application areas. 
In application practice, classically, there is an emphasis on linear system representations, e.g., in the form of experimental modal analysis in mechanical engineering. 
In contrast, the field of nonlinear system identification is, by its nature, even more varied, greatly depending on properties of the specific system~\cite{schoukens2019nonlinear}. 
Lately, not least because of the success of data-driven and machine-learning methods in other fields, there has been a renewed surge in research in data-driven system modeling, whether the focus is primarily on the system identification itself~\cite{schoukens2019nonlinear}, on data-driven surrogate modeling~\cite{kutz2016dynamic}, or other applications, for instance using deep-learning techniques like autoencoders~\cite{beintema2023deep,masti2021learning}. 

A popular application of system identification is to identify models to be used for optimization-based control such as model predictive control~(MPC). 
Notable examples include Gaussian process-based (predictive) control~\cite{bradford2020stochastic,esch21gpmpc,lederer2025episodic}. 
A slightly different perspective is taken by DeePC~\cite{Coulson2019DeePC}, which has become quite popular recently and which integrates the system identification step together with estimation into the optimal control problem to be solved~\cite{willems2005note,faulwasser2023behavioral,chen2026deepc}. 
Still, it has been shown that, at least under specific circumstances, the results can be equivalent as to when identifying the system in a preceding step using the same data, see \cite{fiedler2021subscapedeepc}. 
In any case, when identifying systems for control, it is worth highlighting that control is often used for automation, so controlled systems may operate without human supervision, making it particularly crucial that identified models used for control, and, hence, also the data used in the identification process, are of sufficient quality to enable reliable control performance, especially in safety-critical applications. 
Thus, control-theoretic considerations demand certain properties from identification methods that make them amenable to formal verification of desired closed-loop properties. 
For this, bounds on approximation errors can be particularly helpful, e.g., to ensure asymptotic stability~\cite{schimperna2025data}.  
Moreover, application-side considerations demand answers to the question when enough data is collected, or, put differently, when data is suitable for the learning task~\cite{van2023informativity}. 
Suitability is there often evaluated via rank conditions as, e.g., in persistency of excitation~\cite{willems2005note}, see also~\cite{coulson2022quantitative,alsalti2023design} for first steps towards quantitative notions. 
Such considerations become particularly relevant when trying to identify systems online from measurement data, in the spirit of online- or self-learning, exploration-and-exploitation schemes~\cite{manz20mpclearnonline,esch23exploreexploit}, or adaptive model-based control through adapting data-inferred models, see, e.g., \cite{bold2025two}. 

A particularly important subclass of nonlinear systems consists of control-affine systems as many mechatronic systems like robots and vehicles can be described by control-affine models, and as the control-affinity provides enough structure to arrive at meaningful and narrow-enough conclusions.
For bilinear systems, a subset of this class, the problems of persistent excitation and identifiability have been studied in detail, establishing foundational results on the input signals required for successful identification~\cite{Dasgupta1989,Sontag2009}.
We are concerned with general control-affine structures and their approximation from data as it appears, e.g., when using bilinear extended dynamic mode decomposition with control (bilinear EDMDc) as proposed in~\cite{surana2016koopman,peitz:otto:rowley:2020} or its variants~\cite{philipp2025error} based on kernel extended dynamic mode decomposition (kernel EDMD, \cite{klus2020kernel}).
As shown by~\cite{BoldPhil24,schimperna2025data}, finite-data bounds on the approximation error, which are key for data-driven control with closed-loop guarantees~\cite{strasser2025overview}, depend among others on the interplay of state and control in the available data set.

In this article, we propose a framework for data fitting of control-affine mappings to ensure a desired robustness margin in the associated system identification problem.
To this end, we consider the respective regression problem and derive bounds on the minimal singular value of a matrix composed of the input data---also if bounds on the control inputs are present, see Section~\ref{sec:problem-formulation}. 

In Section~\ref{sec:excitation}, we study conditions under which the input data is \emph{exciting}. In this case, the input-dependent term in the error bound is minimized, or equivalently, the smallest singular value of the data matrix attains its maximum (upper bound). 
This furnishes a necessary optimality condition on the choice of inputs used, see Section~\ref{sec:excitation:optimality} for details.
We show that, under the necessary condition, scaling the input amplitudes by a scalar factor is sufficient for optimality, providing a direct criterion for input design. Using this, we construct inputs that achieve optimality with the fewest data points and propose optimal control inputs for the constrained regression problem.

Whereas these novel contributions are already very useful when identifying a model from scratch, their suitability to certain practical applications can be limited. 
For instance, the aforementioned contributions assume that control inputs can always be appropriately scaled if the optimum shall be attained. 
More crucially, however, it is assumed that the decision on which data to collect is made a priori and jointly for all data points. 
In many practical applications, it is instead desirable to add data sequentially to iteratively refine the system approximation. 
In Section~\ref{sec:excitation:efficient}, we provide a framework for this setting using the concept of subspace angles. 
It turns out that given a set of inputs, choosing an additional input vector so that all inputs together sum to zero, may improve the regression result significantly. Geometrically, this simple strategy centers the inputs and spreads them more symmetrically, removing any bias towards one direction in the input space and thereby improving the conditioning. 

In Section~\ref{sec:Koopman}, the aforementioned theoretical contributions are instantiated for bilinear EDMDc in the Koopman framework to remove potentially restrictive conditions on the data collection enabling flexible sampling. 
The presented results provide constructive, directly implementable insight due to their geometric interpretation, as we exemplarily show for physics- and application-oriented examples in Section~\ref{sec:application}. 
Furthermore, our results are also of value for kernel EDMD and generator EDMD, where the latter is used to learn continuous-time dynamics. 
Regarding kernel EDMD, we provide uniform error bounds on the approximation error in the setting from~\cite{BoldPhil24}, filling an important gap to rigorously verify all assumptions of~\cite{schimperna2025data} w.r.t.\ data-driven MPC in the Koopman framework. 
The article's structure is illustrated in Figure~\ref{fig:graphical_abstract}. 

\begin{figure}
	\centering
	\begin{tikzpicture}[>=triangle 45, scale=0.77   , transform shape]

    \def\yQuestion{-.75cm}
    \def\yResult{-0.75cm}
    \def\xResult{0cm} 
    \def\yDistB{1.25cm}
    \def\xDistB{1.0cm}
    \def\xA{0.25cm}
    \def\xB{.8cm}
    \def\rCirc{0.3cm}

    \definecolor{DescriptionCol}{RGB}{100,100,100}
    \tikzset{
        DescpriptionBlock/.style = {fill=DescriptionCol, minimum width=5cm, minimum height=1.5cm, align=center, fill opacity=0.2, draw = black, draw opacity=1, text=black, text opacity=1},
        DescpriptionBlockB/.style = {fill=white, minimum width=4cm, 
        align=center, fill opacity=0.2, draw = black, draw opacity=1, text=black, text opacity=1},
        ResultsBlock/.style = {minimum width=5.2cm, minimum height=2.25cm, align=center, draw = blue, fill=white, text=black, text opacity=1},
        ResultsBlockB/.style = {fill=blue,  align=center, fill opacity=0.1, draw = none, minimum height=.5cm, text=black, text opacity=1},
        arrowsAbstract/.style = {->, line width =0.3mm, draw = black},
    }

    \node [draw = black,
        DescpriptionBlock, 
        anchor = north east,
        xshift = \xB
    ] (errorBound) 
    {Data-driven Koopman surrogate models \\
    for control-affine systems using \\ 
    extended dynamic mode decomposition (EDMD)
	};


    \node [DescpriptionBlock,
        anchor = north west,
        xshift = \xB,
    ] at (errorBound.north east)  (subspace) 
    {Data-fitting framework for affine-linear maps:\\
    error bounds, robustness, and optimality};
    %

    \coordinate (Sum) at ($ (errorBound.south east) + (0.45*\xB, -2.0cm) $);

    \node[text=black] at (Sum) {\Large $+$};
    \draw[line width = 0.3mm, draw = black] (Sum) circle (\rCirc);

    \node [ResultsBlock,
        anchor = north east,
        xshift = -1cm,
        yshift = -.5cm,
        minimum width=6.7cm
    ] at (Sum)  (refined) 
    {Robust $L^\infty$-error bounds for \\bilinear Koopman control\\ \\};
    \node [ResultsBlockB,
        anchor = south,
        align=center,
    ] at (refined.south)  (errorskEDMD) 
    {\footnotesize kernel EDMD: Thm.~\ref{thm:control_main}};
    \node [ResultsBlockB,
        anchor = south,
        align=center,
    ] at (errorskEDMD.north)  (errorsbEDMD) 
    {\footnotesize EDMD \& generator EDMD: 
    Rem.~\ref{rem:EDMD:operator},
    Prop.~\ref{cor:approx:error:generator}};

    \node [ResultsBlock,
        anchor = north west,
        xshift = 0.7cm,
        yshift = -.5cm,
        minimum width = 7.1cm,
    ] at (Sum) (guideline) 
    {Data-collection guidelines\\ for robust data fitting 
    \\ \\};
    \node [ResultsBlockB,
        anchor = south,
        align=center,
    ] at (guideline.south) (mainT) 
    {\footnotesize  subspace angles: Thm.~\ref{t:main}};
    %
    \node [ResultsBlockB,
        anchor = south,
        align=center,
    ] at (mainT.north) (NCOb) 
    {\footnotesize optimality conditions: Lem.~\ref{lem:NCO}, Cor.~\ref{cor:scaling}, Prop.~\ref{prop:simplex}};
    %

    \draw[arrowsAbstract] (errorBound.south) |- ($(Sum.west) +(-\rCirc,0)$);
    \draw[arrowsAbstract] (subspace.south) |- ($(Sum.east) +(\rCirc,0)$);
    \draw[arrowsAbstract] ($(Sum.south) +(0,-\rCirc)$) |- (refined.east);
    %
    \draw[arrowsAbstract] ($(subspace.south) +(0.65cm,+0.0cm)$) -| ($(guideline.north) +(0cm,0)$);

    %
    \node[anchor=south west] at (errorBound.north west) {Motivation};
    \node[anchor= south west, rotate=90, DescriptionCol] at (errorBound.south west) {\footnotesize Sec.~\ref{subsec:Koopman}};

    %
    \node[anchor=south east] at (subspace.north east) {Approach};
    \node[anchor= south west, rotate=90, , DescriptionCol] at (subspace.south west) {\footnotesize Sec.~\ref{sec:problem-formulation} \& \ref{sec:excitation}};

\end{tikzpicture}
	\caption{Graphical abstract illustrating the motivational background 
    and the proposed framework for data-fitting of control-affine mappings.}
	\label{fig:graphical_abstract}
\end{figure}

\textbf{Notation}: For a matrix $A \in \mathbb R^{m\times n}$, the range and kernel are denoted by $\operatorname{ran}(A)$ and $\operatorname{ker}(A)$, respectively. 
Its spectral, Frobenius, and maximum norms are denoted by $\lVert A\rVert_2$, $\lVert A\rVert_F$, and $\lVert A\rVert_\mathrm{max}:=\max_{i,j} \lvert A_{i,j}\rvert$, respectively. Note that $\lVert A\rVert_\mathrm{max}\leq \lVert A\rVert_2\leq \lVert A\rVert_F$, cf.\ \cite[p.~56]{golub2013matrix}. The Moore-Penrose inverse of $A$ is denoted by $A^\dagger$. We denote the smallest singular value of $A$ by $\sigma_\mathrm{min}(A)$. Given a symmetric matrix $A\in \mathbb R^{r\times r}$, its eigenvalues in nondecreasing order are denoted by $\lambda_1(A)\leq\dots \leq \lambda_r(A)$. Given $x,y\in\mathbb R^n$, $\lVert x\rVert$, $\lVert x\rVert_\infty$, and $\langle x,y\rangle$ denote the Euclidean norm, the maximum norm, and the standard scalar product in $\mathbb R^n$, respectively. 
Let $e_j$ be the $j$th canonical basis vector in $\mathbb R^n$ and $\mathbbm{1}_n = \begin{bmatrix}1 & \dots & 1\end{bmatrix}{}^\top \in\mathbb R^n$. The orthogonal complement of a subset $X\subset\mathbb R^n$ with respect to the standard scalar product is denoted by $X^\bot$. If $X$ is a linear subspace, then $P_X$ denotes the orthogonal projection onto X. Given integers $n,m$ with $0\leq n<m$, we define $[n:m]:=\{n, n+1,\dots, m\}$ and $[m]:=[1:m]$.

\section{Motivation and problem formulation}

This section presents an application of the results derived in Sections~\ref{sec:problem-formulation} and~\ref{sec:excitation}, and serves as a motivation for their development. Therefore, we recap the basics of Koopman theory to consider dynamical systems through the lens of observables. In particular, we highlight that control affinity is preserved for the generator of the Koopman semigroup of linear and bounded operators.
In Section~\ref{sec:Koopman}, we pick up on this application again and show an update of previous found error bounds that were first derived in~\cite{BoldPhil24}.

    \subsection{Koopman operator in dynamical control systems}
    \label{subsec:Koopman}
    In this subsection, we recap the basics of modeling nonlinear (control-affine) systems using Koopman operator theory. 
First, we consider the autonomous nonlinear dynamical system
\begin{align} \label{eq:sys:continuous:autonomous}
    \dot{x}(t) = f(x(t))
\end{align}
with (locally) Lipschitz continuous map $f:\R^n \rightarrow \R^n$ such that local existence and uniqueness of solutions $x(\cdot; \hat{x})$ for the initial value problem consisting of the differential equation~\eqref{eq:sys:continuous:autonomous} and the initial condition $x(0) = \hat{x}$ is ensured.
Assuming global existence for the time being, 
the Koopman semigroup~$(\mathcal{K}^t)_{t \geq 0}$ of bounded linear, but infinite-dimensional operators is defined by the identity
\begin{align}\label{eq:Koopman:identity}
    (\mathcal{K}^t\varphi)(\hat{x}) = \varphi(x(t;\hat{x}))
\end{align}
for all real-valued observable functions~$\varphi \in L^2(\R^n, \R)$, $t \geq 0$, and $\hat{x} \in \mathbb{R}^n$, see, e.g., \cite[pp.\ 3-33]{mauroy:mezic:susuki:2020} or~\cite[Proposition 3.4]{philipp2025error}. 
Identity~\eqref{eq:Koopman:identity} states that, instead of evaluating the observable~$\varphi$ at the solution~$x(\cdot;\hat{x})$ at time~$t$, the Koopman operator can be applied to the observable, and then the resulting function~$\mathcal{K}^t\varphi$ can be evaluated at the initial state~$\hat{x}$. 
The corresponding Koopman generator~$\mathcal{L}$ of this semigroup is defined by 
\begin{align}\label{eq:generator:limit}
    \mathcal{L}\varphi = \lim_{t \searrow 0} \frac{\mathcal{K}^t \varphi - \varphi}{t}
\end{align}
on its domain $\mathcal{D}(\mathcal{L})$, which is 
the set of observables~$\varphi$, for which the limit~\eqref{eq:generator:limit} exists w.r.t.\  
the $L^2$-norm. 
Similarly to the identity~\eqref{eq:Koopman:identity}, the Koopman generator~$\mathcal{L}$ fulfills 
the identity $\mathcal{L} \varphi = \langle \nabla \varphi, f \rangle$, which can be analogously defined for control-affine systems of the form
\begin{align}\label{eq:sys:continuous}
    \dot{x}(t) = f(x(t), u(t)) =  g_0(x(t)) + \sum\nolimits_{k = 1}^m g_k(x(t))\,u_k(t)
\end{align}
and a control function $u(t) \equiv u \in \mathbb{R}^n$ using the concept of Koopman control family~\cite{haseli2025two}, 
where the Koopman operator is 
denoted by $\mathcal{K}^t_u$, $t \in [0, \infty)$. Then, due to linearity of the scalar product, i.e., 
\begin{align*}
    (\mathcal{L}^u \varphi)(x) = \langle \nabla \varphi(x), g_0(x) \rangle + \sum\nolimits_{k = 1}^m \langle \varphi(x), g_k(x) \rangle u_k,
\end{align*}
the Koopman generator~$\mathcal{L}^u$ preserves control affinity, i.e., we have $\mathcal{L}^u = \mathcal{L}^0 + \sum\nolimits_{k = 1}^m (\mathcal{L}^k - \mathcal{L}^0) u_k$, where $\mathcal{L}^0$ and $\mathcal{L}^k$, $k \in \{1, \dots, m\}$, are the generators of the Koopman semigroups $(\mathcal{K}^t_0)_{t \geq 0}$ and $(\mathcal{K}^t_{e_k})_{t \geq 0}$ with inputs $u_0 = 0$ and $u_k = e_k$, respectively, with $e_k$ being the $k$-th unit vecto.
Hence, we only require $m+1$~Koopman generators $\mathcal{L}^{u_k}$, $k \in [0:m]$, to compute the Koopman generator~$\mathcal{L}^u$ and can, thus, ensure linear scaling w.r.t.\ the number of inputs. 
However, while theoretically advantageous, this approach might be rather demanding w.r.t.\ data collection: For each element $x_i$ 
of the set $\mathcal{X} \coloneqq \{x_1, \dots, x_d\} \subset \mathbb{X}$, where $\mathbb{X} \subset \R^n$ is compact with fill distance~$h_\mathcal{X} \coloneqq \sup_{x \in \mathbb{X}}\min_{x_i \in \mathcal{X}} \|x - x_i\|$,
one must acquire the $m+1$ samples $\{(x_i,f(x_i,u_k) \mid k \in [0:m]\}$ in order to generate a data-driven surrogate model meaning sampling precisely at~$x_i$ using the $m+1$ control values $u_0 = 0$ and $u_k = e_k$, $k \in [m]$. 
To alleviate the data collection, \cite{BoldPhil24} proposed a preparatory step to approximately infer the sought-after data from 
$\{(x_{ik},f(x_{ik},u_{ik})) \mid k \in [0:m]\}$, $i \in [d]$, where $x_{ik} \in \mathcal{B}_{r_\mathcal{X}}(x_i) = \{ x \in \mathbb{R}^n | \|x-x_i\| \leq r_{\mathcal{X}}\}$ for some radius $r_{\mathcal{X}} > 0$ and $u_k \in \mathbb{U}$ for some compact, convex set~$\mathbb{U} \subset \mathbb{R}^m$ containing the origin in its interior. 
To this end, $d$~regression problems using the input-data matrices
$$
    V_i = \begin{bmatrix}
            1 & \cdots & 1\\
            u_{i0} & \cdots & u_{id_i} 
    \end{bmatrix}, \qquad i \in [d],
$$
with $d_i \geq m$ are solved, which leads, e.g., to the  error bound
\begin{equation}\nonumber
        \| F(x,u) - F^\varepsilon(x,u)\|_\infty \leq C_1 h_\mathcal{X}^{k+1/2} + C_2(\mathcal{X)} \cdot 
        r_\mathcal{X}
        \max_{i \in [1:d]}\|V_i^\dagger\| \qquad \text{for all }\,(x,u) \in \mathbb{X} \times \mathbb{U}
\end{equation}
for the Euler discretization $F(x, u) \coloneqq x + \Delta t f(x, u)$ with time step~$\Delta t$ of the system dynamics~\eqref{eq:sys:continuous} and its kernel EDMD approximant $F^\varepsilon(x, u)$, 
fill distance $h_\mathcal{X} \leq h_0$ and $r_{\mathcal{X}} < h_{\mathcal{X}}/2$ with constants~$C_1$, $C_2(\mathcal{X})$ and $h_0 > 0$; see \cite[Thm.\ 3, Prop.\ 2]{schimperna2025data} and Section~\ref{subsec:kernel:EDMD} for an in-depth discussion.
While for a small (and fixed) number of indices, a rank condition may be sufficient to infer a 
bound on the term $\max_{i \in [1:d]}\|V_i^\dagger\|$, it is key to establish a uniform bound. 
The underlying reason is that the cardinality~$d$ of the set~$\mathcal{X}$ increases for diminishing 
fill distance, which is needed to ensure a sufficiently tight bound for the first summand of the error bound.  
The next section abstracts this problem further in order to derive a general framework, which is also applicable to suitably bound the second term of the error bound.

    \subsection{Problem formulation: affine data fitting}
    \label{sec:problem-formulation}
    We consider the local identification problem for a control-affine mapping
\begin{equation}\nonumber
    y = g_0(x) + G(x) u
\end{equation}
at a target state $x \in \mathbb{R}^n$, where 
$u \in \mathbb{U} = \{ u \in \mathbb{R}^m \mid \lVert u \rVert \leq r_u \} \subseteq \mathbb R^m$, $r_u>0$, and $y \in \mathbb{R}^{n}$ denote the control input and the output, respectively. 
The functions $g_0: \mathbb R^n\rightarrow \mathbb R^n$ and $G:\mathbb R^n\rightarrow \R^{n\times m}$ are assumed to be locally Lipschitz continuous. 
Instead of being able to measure the output~$y_j$ corresponding to a given input $u_j\in\mathbb U$, $j \in \{0,1,\ldots,\dd\}$, the observed output is subject to a bounded disturbance, i.e.,
\begin{equation}\label{eq:disturbance}
    y_j = g_0(x) + G(x)u_j + \varepsilon_j = \begin{bmatrix}
        g_0(x) & G(x)
    \end{bmatrix} \begin{bmatrix}
        1 \\ u_j
    \end{bmatrix} + \varepsilon_j
\end{equation}
with disturbance $\varepsilon_j \in \mathbb{W} := \{ \varepsilon\in \mathbb{R}^{n} \mid \lVert \varepsilon \rVert \leq r_\varepsilon \}$ for some $r_\varepsilon > 0$, e.g., due to measurement noise. 
Another source of corruption may arise from the fact that consecutive measurements for different control inputs may not be taken at exactly the same state~$x$, but rather at $x_j = x+\delta_j$, resulting in 
    $y_j = g_0(x_j) + G(x_j)u_j$,
where \eqref{eq:disturbance} holds with $\varepsilon_j = g_0(x_j) - g_0(x) + \bigl(G(x_j) - G(x)\bigr) u_j$. 
If the state deviation is bounded by $\lVert \delta_j \rVert = \lVert x_j - x \rVert \leq r_x$, then $\lVert \varepsilon_j\rVert \leq r_\varepsilon$ holds with $r_\varepsilon := (L_{g_0} + L_{G}r_u)r_x$, where $L_{g_0} \geq 0$ and $L_G \geq 0$ are the Lipschitz constants of~$g_0$ and~$G$ in near~$x$.
\begin{supplementbox}
    \textbf{Affine-linear data fitting: regression problem}. Let $x \in \mathbb{R}^n$ and data pairs $(y_j,u_j)_{j=0}^{\dd}$ satisfying~\eqref{eq:disturbance} be given. 
    Then, the estimators~$\hat{g}_0^\star$ of $g_0(x)$ and~$\widehat{G}^\star$ of~$G(x)$ are given by the solution of the regression problem
    \begin{equation}\label{eq:lsq}
        \operatorname*{minimize}_{[\begin{smallmatrix}
            \hat{g}_0 & \widehat{G}
        \end{smallmatrix}] \in \mathbb R^{n\times (m+1)}}\ \left\lVert Y - \begin{bmatrix}
            \hat{g}_0 & \widehat{G}
        \end{bmatrix} V\right\rVert_F
    \end{equation}
    with data matrices $Y := [y_0\ y_1\ \cdots\, y_{\dd}] \in \mathbb{R}^{n \times (\dd+1)}$ and $V \in \mathbb{R}^{(m+1)\times (\dd+1)}$ defined by
    \begin{equation}\label{eq:V}
        V := \begin{bmatrix}
            \one_{\dd+1}^\top \\ U
        \end{bmatrix} = \begin{bmatrix}
            1 & \cdots & 1\\
            u_0 & \cdots & u_{\dd} 
    \end{bmatrix}.
    \end{equation}
\end{supplementbox}

The norm of the residual can be bounded in terms of the smallest singular value $\sigma_{\min}(V)$ of $V$ and the number of measurements as shown in Proposition~\ref{prop:approx:error}.
\begin{proposition}[Error bound]\label{prop:approx:error}
    Let $x\in\mathbb R^n$ and data pairs $(y_j, u_j)_{j=0}^{\dd}$ satisfying~\eqref{eq:disturbance} be given. 
    Then, if the matrix~$V$ defined by~\eqref{eq:V} has full row rank, the solution $[\hat{g}_0^\star\ \widehat{G}^\star]$ of the regression problem~\eqref{eq:lsq} satisfies the error bound
    \begin{equation}\label{eq:bound:approx_Prop1}
        \left\lVert  \begin{bmatrix}
            g_0(x) & G(x)
        \end{bmatrix} - \begin{bmatrix}
            \hat{g}_0^\star & \widehat{G}^\star \end{bmatrix}\right\rVert  _\mathrm{max} \leq r_\varepsilon  \frac{\sqrt{\dd+1}}{\sigma_{\min}(V)},
    \end{equation}
    where $\sigma_{\min}(V)$ denotes the smallest singular value of the matrix~$V$ satisfying the upper bound $\sigma_\mathrm{min}(V) \leq \sqrt{\dd+1}$.
    If, in addition, input constraints $\lVert u_{i} \rVert = \lVert U e_{i+1} \rVert \leq r_u$, $i \in \{0,1,\ldots,\dd\}$, are present with $r_u < \infty$, the upper bound is given by 
    \begin{equation}\label{eq:bound:control-constraints}
        \sigma_\mathrm{min}(V) \leq \min\left\lbrace\sqrt{\dd+1}, r_u \sqrt{\frac{\dd+1}{m}}\right\rbrace.    
    \end{equation}
\end{proposition}
\begin{proof}
    Recall that the unique solution of the least-square regression problem~\eqref{eq:lsq} satisfies $\left[ \hat{g}_0^\star\  \widehat{G}^\star \right] = Y V^\dagger$,
        where $V^\dagger = V^\top (V V^\top)^{-1}$ is the Moore--Penrose inverse of the matrix~$V$. 
    In particular, $VV^\dagger = I_{m+1}$ holds and the norm $\lVert V^\dagger\rVert_2$ is the reciprocal of $\sigma_{\min}(V)$. 
    
    Defining $E:=[\varepsilon_0\ \varepsilon_1\ \cdots\ \varepsilon_{\dd}]$, one finds
    \begin{equation*}
        \begin{bmatrix}
            \hat{g}_0^\star & G^\star
        \end{bmatrix} - \begin{bmatrix}
            g_0(x) & G(x)
        \end{bmatrix} = 
            Y V^\dagger
         - \begin{bmatrix}
            g_0(x) & G(x)
        \end{bmatrix} VV^\dagger = E V^\dagger.
    \end{equation*}
    The assertion follows with
    \begin{equation*}
        \lVert E V^\dagger\rVert_\mathrm{max}\leq\lVert E V^\dagger\rVert_2 \leq \Vert E\rVert_2 \lVert V^\dagger\rVert_2 \leq \Vert E\rVert_F \lVert V^\dagger\rVert_2
            \leq r_\varepsilon\frac{\sqrt{\dd+1}}{\sigma_{\min}(V)}.
    \end{equation*}
    To show the assertion w.r.t.\ the upper bounds on~$\sigma_{\min}(V)$, consider the singular value decomposition $V=Q^\top \Sigma P$, where $Q=\begin{bmatrix}
        q_1,\dots, q_{m+1}
    \end{bmatrix}\in\mathbb R^{(m+1)\times(m+1)}$, $P\in\mathbb R^{(\dd+1)\times(\dd+1)}$ are orthogonal matrices, and $\Sigma = \begin{bmatrix}
        \operatorname{diag}(\sigma_{1},\dots, \sigma_{m+1}) & 0_{(m+1)\times (\dd-m)}
    \end{bmatrix}$. Therein, $(\sigma_j)_{j=1}^{m+1}$ are the singular values of $V$ in descending order. 
    From the structure of $V$ it is evident that 
    \begin{equation}\label{eq:q1}
        q_1^\top \Sigma P = \one_{\dd+1}^\top.
    \end{equation} The orthogonality of $P$ and $Q$ yields $\sqrt{\dd+1} = \lVert \one_{\dd+1} \rVert = \lVert P^\top \Sigma^\top q_1\rVert = \lVert \Sigma^\top q_1\rVert$ and $\sigma_\mathrm{min}(V)=\sigma_{m+1} \leq \lVert \Sigma^\top q_1\rVert \leq \sigma_1$, showing the claimed upper bound on $\sigma_\mathrm{min}(V)$.
    If, in addition, control constraints are present, Cauchy's interlacing property, see \cite[Theorem~4.3.28]{horn2012matrix} yields
    \begin{equation}
        \sigma^2_\mathrm{min}(V)=\lambda_1(VV^\top) \leq \lambda_1(UU^\top)=\sigma^2_\mathrm{min}(U).\label{eq:cauchy_interlacing}
    \end{equation}
    Moreover,
    \begin{equation}
        \sigma^2_\mathrm{min}(U) \leq \frac{1}{m}\sum_{i=1}^{m} \sigma_i^2(U) =  \frac{1}{m} \lVert U\rVert_F^2 = \frac{1}{m} \sum_{j=1}^{\dd+1} \lVert Ue_j\rVert^2 \leq r_u^2 \cdot \frac{\dd+1}{m}.
    \end{equation}
    Together with~\eqref{eq:cauchy_interlacing} and $\sigma_\mathrm{min}(V)\leq \sqrt{\dd+1}$ this implies \eqref{eq:bound:control-constraints}.
\end{proof}

Note that the rank condition on $V$ implies $\dd \geq m$. 
Finding an a priori bound on $\sqrt{\dd+1}/\sigma_{\min}(V)$ depending on the control inputs is far from obvious. 
\begin{remark}
    In \cite[Remark 4.5 (a)]{BoldPhil24}, a probabilistic bound is derived for the case of inputs $u_0,u_1,\dots, u_{\dd}$ drawn independently, uniformly over a hypercube $[-R,R]^m$ with $R>0$, showing that the probability of $\sqrt{\dd+1}/\sigma_{\min}(V)$ being large decays exponentially as $\dd$ increases, see also Figure~\ref{fig:random}. Under the additional assumption $R\geq \sqrt{3}$, in the given scenario one can infer that $\sqrt{\dd+1}/\sigma_{\min}(V)$ converges in probability to the lower bound $1$ as the number of input samples increases, i.e., as $\dd \rightarrow \infty$.
\end{remark}
\begin{figure}[htb]
    \centering
    \includegraphics[width=0.475\linewidth]{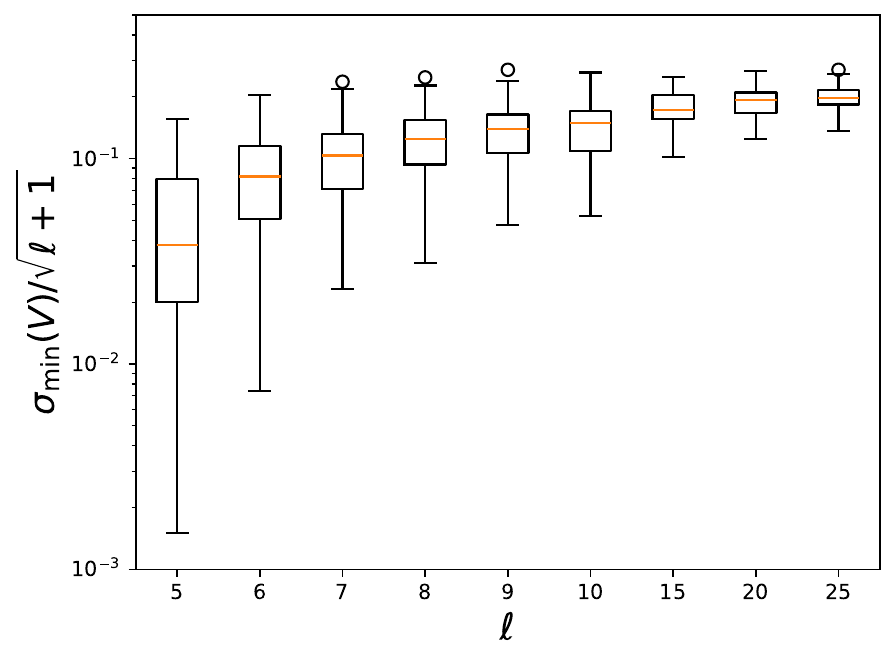}
    \includegraphics[width=0.475\linewidth]{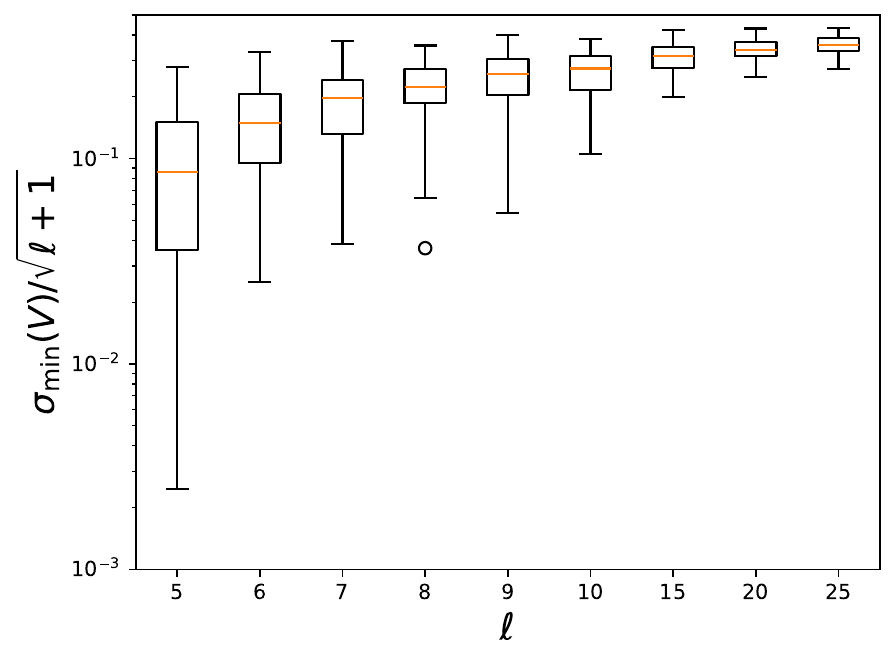}
    \vspace{-1em}
    \caption{Box plots of $\frac{1}{\sqrt{\dd}} \sigma_\mathrm{min}(V)$ for $m = 4$ and $\ell \in \{5, 6, 7, 8, 9, 10, 15, 20, 25\}$ with $u_i$ drawn i.i.d.\ and uniformly from the set $[-0.5, 0.5]^4$ without (left) and with normalization ($\|u_i\| = 1$; right).}
    \label{fig:random}
\end{figure}

\begin{remark}
    In statistics, regression problems of the form~\eqref{eq:disturbance} typically occur in parameter estimations for (affine) linear statistical models, where the errors $\varepsilon_j$ are modeled as centered random variables. In this context $V=\left[\begin{smallmatrix} 1 & \dots & 1\\ u_0 & \dots & u_{\dd}\end{smallmatrix}\right]$ is called the design matrix, and the parameters to be estimated are the intercept and the slope (in our setting $g_0(x)$ and $G(x)$). If the errors are uncorrelated and of equal variance, the Gauss--Markov theorem implies that the best linear unbiased parameter estimator is given by $YV^\dagger$, where $Y=\begin{bmatrix}
        y_0&\dots & y_{\dd}
    \end{bmatrix}$. The structure of $V$, which depends on the choice of the vectors $u_0,\dots, u_{\dd}$, has a strong influence on the quality of the parameter estimation, and different optimality criteria regarding experimental design have been proposed, see, e.g., \cite{pukelsheim2006optimal}. Three classical criteria are D-optimality, which maximizes $\det(VV^\top)$ and minimizes the volume of the confidence ellipsoid, A-optimality, which minimizes the trace $\operatorname{tr}((VV^\top)^{-1})$ and maximizes the average precision, and E-optimality, which maximizes $\sigma_\mathrm{min}(V)$ and provides the best worst-case bound on the estimation error. Among these, D-optimality is by far the most prominent in the literature, while E-optimality has received less attention, and existing results like~\cite{pukelsheim1993optimal,heiligers1994optimal} are restricted to scalar data points $u_j$. To the best of our knowledge, explicit results on E-optimal design like in our setting, see Section~\ref{sec:excitation}, are not available in the literature.
\end{remark}

    \section{On excitation of control-affine systems} 
    \label{sec:excitation}
    In Proposition~\ref{prop:approx:error}, the upper bound consists of two parts. 
On the one hand, it depends on the upper bound~$r_\varepsilon$ on the disturbance resulting from \textit{noisy data}. 
On the other hand, it depends on the quotient $\sqrt{\dd+1} / \sigma_{\operatorname{min}}(V)$. 
In this section, we focus on this quotient or, to be more precise, on its denominator~$\sigma_{\operatorname{min}}(V)$. We develop techniques to generate \emph{exciting} data so that the quotient approaches its lower bound of one. 

In Subsection~\ref{sec:excitation:optimality}, we provide a necessary optimality condition and show that, then, sufficient scaling of the inputs~$u_0,u_1,\ldots,u_{\dd}$ is sufficient to ensure that the lower bound of the quotient $\sqrt{\dd+1} / \sigma_{\operatorname{min}}(V)$ is attained. 
Furthermore, we tighten the upper bound on the minimal singular value $\sigma_{\operatorname{min}}(V)$ if the input and, thus, a potential scaling is constrained, i.e., $r_u < \infty$ holds.

In Subsection~\ref{sec:excitation:efficient}, we mainly focus on the case $\dd = m$ meaning that we are interested in \textit{exciting} inputs. This is of particular interest in many applications when either data collection is expensive, e.g., each data pair corresponds to a costly numerical simulation or even a real-world experiment, or sequentially collected data as, e.g., required in active learning.

\subsection{Necessary and sufficient optimality conditions}
\label{sec:excitation:optimality}

In the following lem\-ma, we provide a necessary condition of optimality~\eqref{eq:NCO}, which will turn out to be very helpful in deriving sufficient conditions and, thus, to ensure that the optimum $\sigma_\mathrm{min}(V) = \sqrt{\dd+1}$ is attained.
\begin{lemma}
\label{lem:NCO}
    Let $\dd \geq m \geq 1$ and suppose $U\in \R^{m\times (\dd+1)}$ has full row rank. 
    Then, if the smallest singular value~$\sigma_\mathrm{min}(V)$ of the matrix~$V$ defined by~\eqref{eq:V} attains its maximum, i.e., $\sigma_\mathrm{min}(V) = \sqrt{\dd+1}$, we have
    \begin{equation}\label{eq:NCO}\tag{NCO}
        U \one_{\dd+1} = 0.
    \end{equation}
\end{lemma}
\begin{proof}
    Consider the singular value decomposition $V = Q^\top\Sigma P$ as in the proof of Proposition~\ref{prop:approx:error} with $(\sigma_j)_{j=1}^{m+1}$ being the singular values of $V$ in descending order and $q_1$ denoting the first column of the orthogonal matrix $Q$. Let us assume that $\sqrt{\dd+1}=\sigma_{m+1}=\sigma_\mathrm{min}(V)$ holds, i.e., $\sqrt{\dd+1}$ is the smallest eigenvalue of $\Sigma\Sigma^\top$. We show that $q_1\in K:=\operatorname{ker}\!\left(\Sigma\Sigma^\top-(\dd+1) I_{m+1}\right)$. 
    
    Assume the contrary, that is, $(I-P_K)q_1\neq 0$, where $P_K$ is the orthogonal projection onto $K$. Then, using $1=\lVert q_1\rVert^2 = \lVert P_K q_1 \rVert^2 + \lVert (I-P_K)q_1\rVert^2$ and $\lVert \Sigma^\top q_1\rVert^2 = \dd+1$, cf.\ the proof of Proposition~\ref{prop:approx:error}, we get
    \begin{equation*}
        1 <  \frac{\lVert \Sigma^\top P_K q_1\rVert^2 + \lVert \Sigma^\top (I-P_K)q_1\rVert^2}{\dd+1} = \frac{\lVert \Sigma^\top q_1\rVert^2}{\dd+1} = 1;
    \end{equation*}
    a contradiction. Therefore, $\Sigma\Sigma^\top q_1 = (\dd+1) q_1$. By \eqref{eq:q1}, $P\one_{\dd+1} = \Sigma^\top q_1$ and
    \begin{equation}\nonumber
        V\one_{\dd+1} = Q^\top \Sigma P\one_{\dd+1} = Q^\top \Sigma\Sigma^\top q_1 = (\dd+1) Q^\top q_1 = \begin{bmatrix}
            1\\ 0_m
        \end{bmatrix}
    \end{equation}
    implying $U\one_{\dd+1}=0$.
\end{proof}

The next proposition shows that a proper scaling of the input matrix~$U$ attains the upper bound $\sigma_\mathrm{min}(V) = \sqrt{\dd+1}$, achieving \textit{optimal excitation}. 
\begin{proposition}[Excitation by scaling]\label{prop:SCO}
    Let $\dd \geq m \geq 1$ and suppose $U\in \R^{m\times (\dd+1)}$ has full row rank. Let Condition~\eqref{eq:NCO} hold, i.e., $U \one_{\dd+1}=0$, then for every scaling factor $\alpha \geq \frac{\sqrt{\dd+1}}{\sigma_\mathrm{min}(U)}$ the scaled matrix
    \begin{equation}\label{eq:Valpha}
        V_\alpha = \begin{bmatrix}
            \one_{\dd+1}^\top\\ \alpha U
        \end{bmatrix}
    \end{equation}
    satisfies $\sigma_\mathrm{min}(V_\alpha)=\sqrt{\dd+1}$.
\end{proposition}
\begin{proof}
    Suppose $U\one_{\dd+1}=0$. Then $V_\alpha$ in \eqref{eq:Valpha} satisfies
    \begin{equation}
        V_\alpha V_\alpha^\top = \begin{bmatrix}
            \dd+1 & 0\\
            0 & \alpha^2 UU^\top
        \end{bmatrix}.
    \end{equation}
    Therefore, $\sigma_\mathrm{min}(V_\alpha) = \min\{\sqrt{\dd+1}, \alpha\sigma_\mathrm{min}(U)\}$ which proves the claim.
\end{proof}
Next, we show a corollary fully exploiting the results of Proposition~\ref{prop:SCO} to construct an optimal input; essentially only using the $m+1$ inputs $u_0,u_1,\ldots,u_m$, see Figure~\ref{fig:scaling} for an illustration.
\begin{figure*}[htb]
    \centering
    \begin{minipage}[b]{0.49\textwidth}
        \centering
        \input{figTikz/sphere_header_orthonormal}
\begin{tikzpicture}[, >=triangle 45]
    \draw[black, fill=none, very thick] (0,0) circle (\Rcut) ;
    \draw[dashed, very thin] (0,-\Rcut) -- (0,0);
    \draw[very thin, ->] (0,0) -- (0,1.2*\Rcut) node[pos=1.0, right, anchor=west] {$u^{(2)}$};
    \draw[dashed, very thin] (-\Rcut,0) -- (0,0);
    \draw[very thin, ->] (0,0) -- (1.2*\Rcut,0) node[pos=1.0, below, anchor=north] {$u^{(1)}$};

    \def\rArc{0.4}
    \begin{scope}[rotate = \alphaA]
        \draw[gray] (\rArc,0) arc[start angle=0, end angle={\alphaB-\alphaA}, radius=\rArc]; 
        \draw[gray, fill=gray] (0.4*\rArc,0.4*\rArc) circle (0.02) ;
    \end{scope}

    \begin{scope}[rotate = \alphaA]
        \draw [->,thick, blue] (0,0) -- (\Rcut,0) node [pos=1.0, right] {$u_1$};
    \end{scope}
    \begin{scope}[rotate = \alphaB]
        \draw [->,thick, red] (0,0) -- (\Rcut,0) node [pos=1.0, above] {$u_2$};
    \end{scope}
        \draw [->,thick, orange] (0,0) -- (\uNx,\uNy) node [pos=1.0, below] {$u_0$};

    \def\lc{0.2}
    \draw[thin] (0,-0.5*\lc) -- ++(0, \lc);
    \draw[thin] (-0.5*\lc, 0) -- ++(\lc, 0);

    \node at (0,1.6*\Rcut) {Corollary~\ref{cor:scaling}};
    \node[anchor=west] at (0.2, -1.9) {$\alpha$-sphere in $\mathbb{R}^2$};

    \phantom{
        \node at (0,-2) {};
    }
    
\end{tikzpicture}
    \end{minipage}
    \begin{minipage}[b]{0.49\textwidth}
        \centering
        \input{figTikz/sphere_header_simplex}
\begin{tikzpicture}[, >=triangle 45]
    \draw[black, fill=none, very thick] (0,0) circle (\Rcut) ;
    \draw[dashed, very thin] (0,-\Rcut) -- (0,0);
    \draw[very thin, ->] (0,0) -- (0,1.2*\Rcut) node[pos=1.0, right, anchor=west] {$u^{(2)}$};
    \draw[dashed, very thin] (-\Rcut,0) -- (0,0);
    \draw[very thin, ->] (0,0) -- (1.2*\Rcut,0) node[pos=1.0, above, anchor=south] {$u^{(1)}$};

    \def\rArc{0.75}
    \begin{scope}[rotate = \alphaA]
        \draw[gray,<->] (\rArc,0) arc[start angle=0, end angle={\alphaB-\alphaA}, radius=\rArc] node[pos=0.5, right]{$\theta$};
    \end{scope}
    \begin{scope}[rotate = \alphaB]
        \draw[gray,<->] (\rArc,0) arc[start angle=0, end angle={\alphaC-\alphaB}, radius=\rArc] node[pos=0.5, above, xshift=-0.2cm]{$\theta$};
    \end{scope}
    \begin{scope}[rotate = \alphaC]
        \draw[gray,<->] (\rArc,0) arc[start angle=0, end angle={120}, radius=\rArc] node[pos=0.6, below]{$\theta$};
    \end{scope}

        \draw [->,thick, blue] (0,0) -- (\uAx, \uAy) node [pos=1.0, right] {${u}_1$};
        \draw [->,thick, red] (0,0) -- (\uBx, \uBy) node [pos=1.0, above] {$u_2$};
        \draw [->,thick, orange] (0,0) -- (\uCx, \uCy) node [pos=1.0, below, anchor=north] {${u}_0$};

    \def\lc{0.2}
    \draw[thin] (0,-0.5*\lc) -- ++(0, \lc);
    \draw[thin] (-0.5*\lc, 0) -- ++(\lc, 0);

    \node at (0,1.6*\Rcut) {Proposition~\ref{prop:simplex}};

    \phantom{
        \node at (0,-2.55) {};
    }
    
\end{tikzpicture}
    \end{minipage}
    \vspace{-1em}
    \caption{Left: Illustration of the choice of orthogonal input vectors proposed in Corollary~\ref{cor:scaling}, where $u^{(j)}$ denotes the respective direction in the input space. 
    Right: Illustration of the simplicial choice of input vectors considered in Proposition~\ref{prop:simplex}, where the angle $\theta$ between the vectors equals $120^\circ$.
    Both choices yield an optimal excitation for the depicted case of $m=2$.}
    \label{fig:scaling}
\end{figure*}
\begin{corollary}[Orthogonal inputs]\label{cor:scaling}
    Choose $U=\alpha\begin{bmatrix}
        u_0&\dots&u_\dd
    \end{bmatrix}$ in~\cref{eq:V} with a scaling factor $\alpha\geq\sqrt{\dd+1}$ such that \eqref{eq:NCO} is satisfied, 
    the input vectors $u_1,\dots,u_m$ form an orthonormal basis of $\R^m$, and $u_j=0$ holds for all $j \in \{m+1,\ldots,\dd\}$.
    Then, we have $\sigma_\mathrm{min}(V)=\sqrt{\dd+1}$.
\end{corollary}
\begin{proof}
    Per assumption, $U_m:=\begin{bmatrix}
        u_1 & \dots &u_{m}
    \end{bmatrix}$ is an orthogonal matrix and
    \begin{equation}\nonumber
     V=\begin{bmatrix}
         1 & \one_{m}^\top & \one_{\dd-m}^\top\\
         -\alpha U_m\one_m & \alpha U_m & 0
     \end{bmatrix}.
    \end{equation}
    A direct computation yields
    \begin{equation}\nonumber
         VV^\top = \begin{bmatrix}
            (\dd+1) & 0\\0 & \alpha^2 (U_m\one_m \one_m^\top U_m^\top + I_m)
        \end{bmatrix}.
    \end{equation}
    From the structure of $VV^\top\!$, one finds
    \begin{equation}\nonumber
        VV^\top e_1 = (\dd+1)e_1, \quad VV^\top \begin{bmatrix}
            0\\ U_m\one_m
        \end{bmatrix} = \alpha^2(m+1) \begin{bmatrix}
            0\\ U_m\one_m
        \end{bmatrix}
    \end{equation}
    and for $x\in \{U_m\one_m\}^\bot$ 
    \begin{equation}\nonumber
        VV^\top \begin{bmatrix}
            0\\ x
        \end{bmatrix} = \alpha^2 \begin{bmatrix}
            0\\ x
        \end{bmatrix}.
    \end{equation}
    This shows that the eigenvalues of $VV^\top$ are given by $\lambda_1=\dd+1$, $\lambda_i = \alpha^2$, $i \in \{2,\ldots,m\}$, and $\lambda_{m+1} = \alpha^2(m+1)$ showing the assertion.
\end{proof}
Indeed, increasing the scaling of the control inputs raises the smallest singular value, up to the point where it eventually saturates or the bound~$r_u$ pertaining to the control constraint set~$\mathbb{U}$ leads to a saturation.

\begin{remark}[Balanced normalized tight frames: BNTF]
    The problem of maximizing the smallest singular value of $V$ is closely related to the topic of frames, see~\cite{heineken2020balanced} for a discussion of existence and construction of BNTFs.
    A sequence $\{u_j\}_{j=0}^{\dd}\subset \mathbb R^m$ is a balanced normalized tight frame (BNTF) if it satisfies the balancing condition $\sum_{j=0}^{\dd} u_j = 0$, each element is normalized, $\lVert u_j\rVert=1$, $j\in\lbrace0,\dots,\dd\rbrace$, and the tightness property holds, i.e., $\sum_{j=0}^{\dd} \langle x,u_j\rangle^2 = A\lVert x\rVert^2$ for all $x\in\mathbb R^m$ with $A=\frac{\dd}{m}$. Hence the frame operator is $U_FU_F^\top = \frac{\dd+1}{m} I_{m}$, $U_F:=\begin{bmatrix}u_0 & \dots & u_{\dd}\end{bmatrix}$, see \cite[Theorems~2.1,~3.1]{benedetto2003finite}. If $U=\alpha U_F$ rescales the frame by a scaling factor $\alpha>0$, the matrix $V$ in \eqref{eq:V} satisfies
    \begin{equation*}
        \sigma_\mathrm{min}(V) = \min\{\sqrt{\dd+1}, \alpha\sigma_\mathrm{min}(U_F)\} = \min\!\left\lbrace\!\sqrt{\dd+1},\ \alpha\sqrt{\tfrac{\dd+1}{m}}\right\rbrace.
    \end{equation*}
    Choosing $\alpha=r_u$, ensures that input constraint is respected and in this case $\sigma_\mathrm{min}(V)$ exactly attains the bound in~\eqref{eq:bound:control-constraints} of Proposition~\ref{prop:approx:error}. 
    Moreover, for $\alpha \geq \sqrt{\tfrac{m}{\dd+1}}$, one reaches the sharp bound $\sigma_\mathrm{min}(V)=\sqrt{\dd+1}$.
\end{remark}

\subsection{Sequential data collection using subspace angles} 
\label{sec:excitation:efficient}

Before we proceed, we show that exciting the system using simplex vertices yields an optimal solution of the input-constrained regression problem~\eqref{eq:lsq} in order to motivate the following results, see Figure~\ref{fig:scaling} (right) for an illustration. Furthermore, the results show, in the case $m=\dd$, that the bound in~\eqref{eq:bound:control-constraints} of Proposition~\ref{prop:approx:error} is sharp.
\begin{proposition}[Simplex vertices as inputs]\label{prop:simplex}
    Let the first $m+1$ control inputs $u_0,u_1,\ldots,u_m$ be the vertices of a regular $m$-simplex, i.e., 
    \[
        u_0  = -\frac{1}{\sqrt{m}} \one_{m}, \quad u_j = \sqrt{\frac{m+1}{m}} e_j + \frac{1-\sqrt{m+1}}{m\sqrt{m}} \one_{m}, \ j \in \{1,\dots,m\},
    \]
    and set $u_j = 0$ for all $j \in \{ m+1,\ldots,\dd\}$. Then, using $U = \alpha\begin{bmatrix}
        u_0&\dots & u_{\dd}
    \end{bmatrix}$ with scaling factor $\alpha>0$ to construct the matrix~$V$ in~\eqref{eq:V}, we have $\sigma_\mathrm{min}(V) = \min\left\lbrace\sqrt{\dd+1}, \alpha \sqrt\frac{m+1}{m}\right\rbrace$, while $\lVert u_j \rVert \leq \alpha$ holds for all $j \in \{0,1,\ldots,\dd\}$.
    In particular, $\sigma_\mathrm{min}(V)$ attains its upper bound~\eqref{eq:bound:control-constraints} for $r_u = \alpha$ and $m = \dd$ for constrained inputs, i.e., $u_j \in \mathbb{U}$, $j \in \{0,1,\ldots,\dd\}$.
\end{proposition}
\begin{proof}
    Let $a = -\frac{1}{\sqrt{m}}$, $b = \sqrt{\frac{m+1}{m}}$, $c = \frac{1-\sqrt{m+1}}{m\sqrt{m}}$, and $t = \sqrt{m+1}$. Then 
    \begin{eqnarray*}
        \langle u_i, u_j \rangle & = & 2bc + mc^2 = 2\frac{t(1-t)}{m^2} + \frac{(1-t)^2}{m^2} = \frac{1-t^2}{m^2} = -\frac{1}{m}, \\
        \langle u_i, u_i \rangle & = & b^2 + 2bc + mc^2 = \frac{m+1}{m} - \frac{1}{m} = 1, \\
        \langle u_0, u_j \rangle & = & ab + mac = - \frac{t}{m} + \frac{t-1}{m} = -\frac{1}{m}
    \end{eqnarray*}
    for $i,j\in\lbrace 1,\dots, m\rbrace$, $j\neq i$.
    In particular, $\lVert u_i \rVert = 1$ and $\lVert U e_{i+1} \rVert = \alpha$ hold for all $i\in\lbrace 0,\dots,m \rbrace$. Therefore,
    \begin{equation}\nonumber
        V^\top V = \one_{\dd+1} \one_{\dd+1}^\top + U^\top U = \one_{\dd+1} \one_{\dd+1}^\top 
        + \alpha^2\begin{bmatrix}
            -\frac{1}{m} \one_{m+1}\one_{m+1}^\top + (1+\frac{1}{m}) I_{m+1} & 0\\ 0 & 0
        \end{bmatrix}
    \end{equation}
    and, consequently, 
    \begin{equation}\nonumber
        V^\top V \one_{\dd+1} = (\dd+1)\one_{\dd+1},\quad V^\top V \one_{\dd+1}\begin{bmatrix}
            x\\ 0
        \end{bmatrix} = \alpha^2\frac{m+1}{m}\begin{bmatrix}
            x\\ 0
        \end{bmatrix}
    \end{equation}
    holds for all $x\in \{\one_{m+1}\}^\bot$, which implies the assertion.
\end{proof}
Whereas Corollary~\ref{cor:scaling} and Proposition~\ref{prop:simplex} provide optimal solutions in the unconstrained ($r_u = \infty$) and input-constrained case ($r_u \in (0,\infty)$), their applicability assumes full flexibility w.r.t.\ the choice of~$U$ --~including scaling. 
Next, we provide a framework providing guidance on constructing a set of \textit{sufficiently exciting} inputs. 
To this end, the concept of subspace angles is leveraged in order to derive a lower bound, which is robust to small deviations. 
The subspace angle is the angle $\theta(y,X)$ between a vector $y\in\R^n$ and a linear subspace $X\subset \R^n$, which is defined by
\begin{equation}
    \cos \theta(y,X) := \min_{x\in X\setminus\{0\}} \frac{\lvert\langle y, x\rangle\rvert}{\lVert y\rVert \lVert x\rVert} = \frac{\lVert P_X y\rVert}{\lVert y\rVert},
\end{equation}
where $P_X$ denotes the orthogonal projection onto the space $X$.
This intuitive concept allows practitioners to directly infer a required accuracy for ensuring a sufficiently large lower bound on~$\sigma_\mathrm{min}(V)$ and thus a small upper bound~\eqref{eq:bound:approx_Prop1}. Exemplary applications follow in the next two sections. %
In conclusion, the previous results are of particular interest for a preparatory offline phase, while the following main result is applicable at runtime.

We focus on the case $\dd=m$, in which $V \in \R^{(m+1)\times (m+1)}$ defined by~\eqref{eq:V} is a quadratic matrix. 
Moreover, we define the matrix $U_m = \begin{bmatrix} u_1 & \cdots & u_m \end{bmatrix}$, i.e., $U$ without the first column.  
A key step in the upcoming analysis is the treatment of the $\one_{m+1}$-vector in the first row of~$V$ and the impact of~$u_0$ on our lower bound on~$\sigma_\mathrm{min}(V)$. To this end, we consider the function
\begin{equation}
\label{eq:theta}
    \Theta : \R^m\to [0,1],\quad x \mapsto 1 - \sqrt{\frac{m+1 - \frac{(1-\one_m^\top x)^2}{1+\|x\|^2}}{m+1}}.
\end{equation}
It will turn out that $1-\Theta(U_m^{-1}u_0)$ coincides with the cosine of the angle. 
Note that $0\le\Theta(x)\le 1$ for all $x\in\R^m$. In particular, the function~$\Theta$ vanishes on the affine subspace $\tfrac 1m\one_m + (\operatorname{span}\{\one_m\})^\perp = \{x\in\R^m : \one_m^\top x = 1\}$, attains its maximum of $1$ at $x=-\one_m$, and satisfies $\Theta(0) = 1 - \sqrt{m/(m+1)}$, see also Figure~\ref{fig:Theta}.
\begin{figure}
    \centering
    \begin{minipage}[t]{0.5\textwidth}
        \centering
        \vspace{0pt}
        \begin{tikzpicture}
\begin{axis}[
    grid=major,
    view={75}{50},
    scale=0.7,
    zmin = 0,
    zmax = 1,
]
\addplot3[
    surf,
    samples=66, 
    domain=-5:5,
    miter limit=1,
    colormap/hot
]
{1-sqrt((3-(1-x-y)^2/(1+x^2+y^2))/3)};

\addplot3[
    domain=-4:5,
    samples=2,
    thick,
    color=red
]
({x}, {1-x}, {0});
\end{axis}
\end{tikzpicture}
    \end{minipage}%
    \begin{minipage}[t]{0.5\textwidth}
        \centering
        \vspace{0pt}
        \input{figTikz/sphere_header_bisector.tex}
\begin{tikzpicture}[, >=triangle 45, scale=1.2]
    \draw[black, fill=none, very thick] (0,0) circle (\Rcut) ;
    \draw[dashed, very thin] (0,-\Rcut) -- (0,0);
    \draw[very thin, ->] (0,0) -- (0,1.3*\Rcut) node[pos=1.0, left, anchor=east] {$u^{(2)}$};
    \draw[dashed, very thin] (-\Rcut,0) -- (0,0);
    \draw[very thin, ->] (0,0) -- (1.3*\Rcut,0) node[pos=1.0, below, anchor=north] {$u^{(1)}$};


    \begin{scope}[rotate = \alphaA]
        \draw [->,thick, blue] (0,0) -- (\Rcut,0) node [pos=1.0, above, anchor=south west] {${u}_1$};
    \end{scope}
    \begin{scope}[rotate = \alphaB]
        \draw [->,thick, red] (0,0) -- (\Rcut,0) node [pos=1.0, left] {${u}_2$};
    \end{scope}
    \begin{scope}[rotate = \alphaC]
        \draw [->,thick, orange] (0,0) -- (2*\Rcut,0) node [pos=1.0, below, anchor = north west] {${u}_0$};
    \end{scope}

    \pgfmathsetmacro{\alphaBisector}{0.5*(\alphaA+\alphaB)}
    \begin{scope}[rotate = \alphaBisector]
        \draw [dashdotted, orange] (0,0) -- (\Rcut,0);
    \end{scope}

    \def\lc{0.2}
    \draw[thin] (0,-0.5*\lc) -- ++(0, \lc);
    \draw[thin] (-0.5*\lc, 0) -- ++(\lc, 0);


    
\end{tikzpicture}
    \end{minipage}
    \vspace{-3em}
    \caption{Left: Surface plot of the function $\Theta$ in \eqref{eq:theta} for the case $m=2$ on the box $[-5,5]^2\subset\R^2$ with peak at $-\one_2 = [-1,-1]^\top$. The affine subspace on which $\Theta$ vanishes is indicated as line red. Right: Choosing $u_0 = -(u_1+u_2)$ for two randomly given $u_1$ and $u_2$ to maximize the subspace angles for $m=2$, see Theorem~\ref{t:main}.}
\label{fig:Theta}
\end{figure}

The following theorem is the main result of this section and yields a geometrically interpretable lower bound on~$\sigma_\mathrm{min}(V)$ in terms of subspace angles, see Figure~\ref{fig:Theta}.
\begin{theorem}\label{t:main}
    Let $U_m$ be invertible and $\{i_1,\ldots,i_m\} = \{1,\dots,m\}$ be such that $\|u_{i_1}\|\ge\|u_{i_2}\|\ge\cdots\ge\|u_{i_m}\|$. Further, set $I_j = \{1,\dots,m\}\backslash\{i_1,\ldots,i_j\} = \{i_s : s>j\}$. Then, if $\theta(u_{i_s},S_{I_s})$ denotes the angle between the vector $u_{i_s}$ and the subspace $S_{I_s}$ defined by $\operatorname{span}\{u_i : i\in I_s\}$, we have
    \begin{align}\label{e:main_estimate}
        \sigma_\mathrm{min}^2(V) \ge \Theta(U_m^{-1}u_0)\,\cdot\,\min\left\{m+1,\,\|u_{i_m}\|^{2}\,\cdot\,\prod_{s=1}^{m-1}\big(1 - \cos\theta(u_{i_s},S_{I_s})\big)\right\}.
    \end{align}
    In particular, if $u_0 = -U_m\one_m$ and $\|u_{i_m}\|\le m+1$ hold, the Inequality~\eqref{e:main_estimate} simplifies to $\sigma_\mathrm{min}^2(V)\ge\|u_{i_m}\|^{2}\,\cdot\,\prod_{s=1}^{m-1}\big(1 - \cos\theta(u_{i_s},S_{I_s})\big)$.
\end{theorem}
\begin{remark}\label{rem:extending_V}
    The assumption $\dd=m$ can be lifted, as adding more data columns of the form $v = \left[\begin{smallmatrix}1\\u\end{smallmatrix}\right]$ to the matrix $V$ can only increase the smallest singular value. Indeed, if $v\in\R^{m+1}$, then
    \begin{align}
    \begin{split}\label{eq:increse:singular:value}
        \sigma_\mathrm{min}^2(\begin{bmatrix}V & v\end{bmatrix})
    &= \lambda_{1}\left(\begin{bmatrix}V & v\end{bmatrix}\begin{bmatrix}
        V\\v
    \end{bmatrix}\right) = \lambda_{1}(V\,V^\top + vv^\top) \\
    &= \inf_{\|x\|=1}x^\top(V\,V^\top + vv^\top)x
    = \inf_{\|x\|=1}x^\top V\,V^\top x + (v^\top x)^2\,\\
    &\ge\,\inf_{\|x\|=1}x^\top V\,V^\top x = \lambda_{1}(V\,V^\top) = \sigma_\mathrm{min}^2(V).
    \end{split}
    \end{align}    
\end{remark}
The proof of Theorem~\ref{t:main} builds in its first step, i.e., the treatment of the first row and the $(m+1)\text{th}$ input~$u_0$, upon the following preparatory lemma, which provides an estimate on the smallest positive eigenvalue for sums of symmetric positive semi-definite (SPSD) matrices. 
To this end, we require the following notation: Given an SPSD matrix $P\in\R^{n\times n} \setminus\{0\}$, 
let $\lambda_\mathrm{min}(P)$ be the smallest positive eigenvalue of $P$, i.e., $\lambda_\mathrm{min}(P) = \lambda_{n-r+1}(P)$ with $r:=\operatorname{rank}(P)>0$.
\begin{lemma}\label{t:kaur_ext}
    Let $u\in\R^n$ and set $P = uu^\top$. Moreover, let $Q\in\R^{n\times n}\setminus\{0\}$ be an SPSD matrix. Then,
    \begin{equation}\label{e:kaur_ext}
    \lambda_\mathrm{min}(P+Q) \geq (1 - \cos\theta(u,\operatorname{ran}Q))\cdot\min\!\left\lbrace \|u\|^2,\,\lambda_{\min}(Q)\right\rbrace.
    \end{equation}
\end{lemma}
\begin{remark}\label{r:kaur}
If $u\in\operatorname{ran}Q$, the statement of the theorem is trivial. If $u\notin\operatorname{ran}Q$, it follows that $Q$ has a non-trivial kernel. In the case where $\ker P\cap\ker Q = \{0\}$ (i.e., $P+Q$ is positive definite), Lemma~\ref{t:kaur_ext} is a special instance of~\cite[Theorem 3.1]{kaur23}. The main contribution of Lemma~\ref{t:kaur_ext} is that it also holds for singular sums $P+Q$ with rank-one matrix $P$. We leave it as an open problem to extend Lemma~\ref{t:kaur_ext} to SPSD matrices $P$ with higher rank.
\end{remark}
The proofs of Theorem~\ref{t:main} and Lemma~\ref{t:kaur_ext} are given in the Appendix, see Section~\ref{sec:appendix}.

    \section{Flexible sampling and uniform error bounds for bilinear EDMD}
    \label{sec:Koopman}

This section presents an application of the results derived in Sections~\ref{sec:problem-formulation} and~\ref{sec:excitation} based on the Koopman theory introduced in Subsection~\ref{subsec:Koopman}, 
which originally motivated their development. 
To this end, 
we show how the affine-linear data fitting from Subsection~\ref{sec:problem-formulation} can be applied such that flexible data sampling for bilinear EDMDc is possible.  
In Subsection~\ref{subsec:kernel:EDMD}, we discuss kernel EDMD, a variant of EDMD using data informed observables to model the underlying dynamics. The contribution in this section is an update of the error bounds for the control extension that was first derived by~\cite{BoldPhil24}.

\subsection{Generator extended dynamic mode decomposition}
\label{subsec:generator:EDMD}

Instead of using linear EDMDc \cite{proctor2016dynamic, korda2018linear} as a learning algorithm to obtain a data-driven surrogate of the system~\eqref{eq:sys:continuous}, we pursue a bilinear approach based on generator EDMD (gEDMD), where the preservation of the control-affine structure is exploited \cite{williams2016extending, surana2016koopman, peitz:otto:rowley:2020}. 

Let $\mathbb{X} \subset \R^n$, $\mathbb{U} \subset \R^m$ be compact, non-empty with the origin in their interior. 
Moreover, let the $M$-dimensional subspace~$\mathbb{V} = \operatorname{span} \{\psi_p \in \mathcal{D}(\mathcal{L}^{u})  \mid p \in \{1, \dots, M\}\}$ be spanned by a dictionary of observable functions and $\Psi \coloneqq (\psi_1, \dots, \psi_M)^\top$ denote the vector-valued observable function where all $M$ observables are stacked. Consider the set $\mathcal{X} = \{x_1, \ldots, x_d\} \subset \mathbb{X}$ and assume that data points are given by 
\begin{align*}
    \Psi(\mathcal{X}) = \{\Psi(x_1), \dots, \Psi(x_d)\} \quad \text{and} \quad \mathcal{L}^k\Psi(\mathcal{X}) = \{(\mathcal{L}^k\Psi)(x_1), \dots, (\mathcal{L}^k\Psi)(x_d)\}
\end{align*}
for all $k \in \{0, \dots, m\}$, where we define $(\mathcal{L}^k\Psi)(x) := ((\mathcal{L}^k\psi_1)(x), \dots, (\mathcal{L}^k\psi_M)(x))^\top $ with $(\mathcal{L}^0\psi_p)(x_j) \hspace*{-1mm}=\hspace*{-1mm} \nabla \psi_p(x_j)^\top g_0(x_j)$ and $(\mathcal{L}^k\psi_p)(x_j) = \nabla \psi_p(x_j)^\top (g_0(x_j) + g_k(x_j))$. 
Assembling the data points in the matrices~$X, \hat{Y}^k \in \R^{M \times d}$ with 
\begin{align}\label{eq:generator:datamatrix}
    X = \Big[\Psi(x_1), \dots, \Psi(x_d)\Big]
    \quad \text{and} \quad 
    \hat{Y}^k = \Big[(\mathcal{L}^k\Psi)(x_1), \dots, (\mathcal{L}^k\Psi)(x_d)\Big],
\end{align}
an approximation of the compressed Koopman generator~$P_{\mathbb{V}}\mathcal{L}^k|_{\mathbb{V}}$ is given by
\begin{align} \label{eq:generator:regression}
    L^k = \argmin_{L \in \R^{M \times M}} \|LX - \hat{Y}^k\|_F^2.
\end{align}
For this proposed bilinear approach on gEDMD, a major disadvantage emerges, namely the need of data points pertaining to a selection of particular (constant) 
control inputs, e.g., the unit vectors of $\R^m$ and $u = 0$. 
Therefore, only certain, specifically crafted data sets can be used. 
The following part uses the method from Section~\ref{sec:problem-formulation} to allow flexible sampling while still obtaining a bilinear gEDMD-based surrogate model.
\\

\textbf{Bilinear gEDMD with flexible sampling.}
Let $\Psi \in \mathcal{C}^1(\mathbb{X}, \R^M)$ be locally Lipschitz-continuous with Lipschitz constant $L_\Psi > 0$. To be able to avoid restrictive sampling, an ideal data set, sufficient to set up and solve a regression problem as in \eqref{eq:generator:regression}, would be of the form
\begin{align}\label{eq:data:ideal:generator} 
    (\psi_p(x_i), \nabla\psi_p(x_{i})^\top g_k(x_i)) , 
\end{align}
for $i \in \{1,\ldots,d\}$, $k \in \{0,\ldots,m\}$, and $p \in \{1, \dots, M\}$. 
 However, typical data sets only contain information about the observables' derivatives along the full dynamics, rather than along its components~$g_k$ that define the dynamics via the control-affine form. 
 Excitation of the system, as proposed in  Section~\ref{sec:problem-formulation}, provides sufficient information such that data of the form \eqref{eq:data:ideal:generator} can be approximated at points~$\mathcal{X} = \{x_1, \dots, x_d\} \subseteq \mathbb{X}$ that do not have to coincide with the sampled data and can be chosen arbitrarily. 
 Assume that the data is given by 
\begin{align}\label{eq:data:clustered}
    (\Psi(x_{ij}),u_{ij},\nabla\Psi(x_{ij})^\top f(x_{ij},u_{ij})) \in \mathcal{B}_{L_\Psi r_{x_i}}(\Psi(x_i)) \times \mathbb{U} \times \mathbb{R}^M
\end{align} 
for $i\in \{1, \dots, d\}$, $j \in \{0, \dots, d_i\}$ with $d_i \geq m$, and cluster radii~$r_{x_i} \geq 0$. 
Here, for each $x_i$ the data pairs~$(u_{ij},\nabla\Psi(x_{ij})^\top f(x_{ij},u_{ij}))_{j = 0}^{d_i}$ correspond to the pairs~$(y_j, u_j)_{j = 0}^d$ for $x$ from Section~\ref{sec:problem-formulation}. Following the proposed structure in Section~\ref{sec:excitation}, we aim to perform the regression~\eqref{eq:lsq} to approximate the points $\nabla\Psi(x_i)^\top g_k(x_i)$. Thereby, we set $V = V_i$ with 
\begin{align}\label{eq:g:V}
    V_i &\coloneqq \Bigg[\begin{array}{c|c|c}
         1 & \cdots & 1  \\
         u_{i1} & \cdots & u_{id_i}
    \end{array}\Bigg], \\
\begin{split}\label{eq:g:YV}
    Y &= \Big[\nabla\Psi(x_{i1})^\top f({x}_{i1}, u_{i1}) \mid  \cdots \mid \nabla\Psi(x_{i1})^\top f({x}_{id_i}, u_{id_i})\Big], 
    \text{ and } \quad \\
    \begin{bmatrix}
    \hat{g}_0 & \hat{G}
\end{bmatrix} &= \Big[ \tilde{Y}^0_i \mid \tilde{Y}^1_i \mid \cdots \mid \tilde{Y}^m_i \Big],
\end{split} 
\end{align}
where $\tilde{Y}^k_i \approx \nabla\Psi(x_{i})^\top{g}_k(x_i)$. We then define $\tilde{Y}^k = \Big[ \tilde{Y}_1^k \mid \cdots \mid \tilde{Y}_d^k \Big]$ and because of the preservation of control-affinity of the Koopman generator, $\tilde{Y}^k$ is an approximation of ${Y}^k \coloneqq \hat{Y}^k - Y^0$ for $k \in \{1, \dots, m\}$, where $Y^0 \coloneqq \hat{Y}^0$ and $\hat{Y}^k$ from~\eqref{eq:generator:datamatrix}. 

\begin{proposition}\label{cor:approx:error:generator}
    Let $\Psi \in \mathcal C^1(\mathbb{X}, \R^M)$ be an observable function and let its Jacobian matrix~$\nabla\Psi$ be locally Lipschitz continuous on $\mathbb{X}$. Further, let $\mathcal{X} = \{x_1, \dots, x_d\} \subset \R^n$ and data according to~\eqref{eq:data:clustered} be given such that $V_i \in \R^{(m + 1) \times d_i}$
    has a full row rank, i.e., $\operatorname{rank}(V_i) = m + 1$ for $i \in \{1, \dots, d\}$. Moreover, let the control inputs be arranged such that $U_m = \Big[u_{i1} \mid \cdots \mid u_{im}\Big]$ is invertible and let $\{\iota_1,\ldots,\iota_m\} = \{1, \dots, m\}$ such that $\|u_{\iota_1}\|\ge\cdots\ge\|u_{\iota_m}\|$, and set $I_j = \{1, \dots, m\}\backslash\{\iota_1,\ldots,\iota_j\} = \{\iota_p : p>j\}$. 
    Based on this, the subspaces~$S_{I_s}$ are defined by $ S_{I_s}= \operatorname{span}\{u_{ij} : j\in I_s\}$. Then, the solution~$\Big[ \tilde{Y}^0_i \mid \tilde{Y}^1_i \mid \cdots \mid \tilde{Y}^m_i \Big]$ of the linear regression problem~\eqref{eq:lsq} with parameters~\eqref{eq:g:YV} satisfies the error bound~\eqref{eq:bound:approx_Prop1}, i.e., 
    \begin{equation}\label{eq:bound:approx_Prop4.1}
        \left\lVert  Y^k - \tilde{Y}^k\right\rVert_\mathrm{max} 
        < r_\varepsilon \max_{i \in \{1, \dots, d\}}\sqrt{d_i} \tilde{\sigma}_i
    \end{equation}
    with $\tilde{\sigma}_i \coloneqq \Big(\Theta(U_m^{-1}u_{i0}) \cdot \min\{m + 1, \|u_{i\iota_m}\| \prod_{s = 1}^{m - 1}(1 - \cos\theta(u_{i\iota_s}, S_{I_s}))\}\Big)^{-\frac12}$  and $r_\varepsilon \coloneqq (L_{\Psi_{g_0}} + L_{\Psi_G} r_u)\max_{i \in \{1, \dots, d\}}r_{x_i}$ depending on constants~$L_{\Psi_{g_0}}, L_{\Psi_G} > 0$.
\end{proposition}

\begin{proof}
As $\nabla \psi_p$ and $g_k$ are Lipschitz-continuous functions for all $p \in \{1, \dots, M\}$ and $k \in \{0, \dots, m\}$, the term~$\nabla \Psi \cdot g_k$ again is Lipschitz-continuous on $\mathbb{X}$. We denote the Lipschitz-constant of $\nabla \Psi \cdot g_0$ as $L_{\Psi_{g_0}}\coloneqq L_{\Psi_{g_0}}(\mathbb{X})$ and the maximum of $\nabla \Psi \cdot g_k$ for $k \in \{1, \dots, m\}$ as $L_{\Psi_G} \coloneqq L_{\Psi_G}(\mathbb{X})$. 
Then, applying Proposition~\ref{prop:approx:error} yields an error bound 
\begin{align} \label{eq:error:}
         \left\lVert \Big[ {Y}^0_i \mid {Y}^1_i \mid \cdots \mid {Y}^m_i \Big] - \Big[ \tilde{Y}^0_i \mid \tilde{Y}^1_i \mid \cdots \mid \tilde{Y}^m_i \Big]\right\rVert_\mathrm{max} 
        < r_\varepsilon \frac{\sqrt{d_i}}{\sigma_{\min}(V_i)} 
    \end{align}
with $r_\varepsilon := (L_{\Psi_{g_0}} + L_{\Psi_G}r_u)r_x$.
The term $\sigma_{\min}(V_i)$ can now be addressed using our findings from Section~\ref{sec:excitation}: Due to the full row rank, we can assume that the matrix $V_i$ can be written as $V_i = \Big[\tilde{V}_i \mid \bar{V}_i\Big]$ with matrix~$\tilde{V}_i = \begin{bmatrix}\mathbbm{1}_{m+1}^\top\\U_{m+1}\end{bmatrix} \in \R^{m + 1\times m + 1}$ with $U_{m+1} = \big[ {u}_{i0} \mid \cdots \mid {u}_{im} \big]$ such that $U_m \coloneqq \big[ {u}_{i1} \mid \cdots \mid {u}_{im} \big]$ is invertible. 
Following Remark~\ref{rem:extending_V}, we find $\sigma_\mathrm{min}(V_i) \geq \sigma_\mathrm{min}(\tilde{V_i})$ 
and thus with Theorem~\ref{t:main} 
\begin{align*}
   \sigma_{\min} (V_i)^{-1} &\leq \sigma_\mathrm{min}(\tilde{V}_i)^{-1} \\
    &\leq \Big(\Theta(U_m^{-1}u_{i0}) \cdot \min\{m + 1, \|u_{i\iota_m}\| \prod_{s = 1}^{m - 1}(1 - \cos\theta(u_{i\iota_s}, S_{I_s}))\}\Big)^{-\frac12}\eqqcolon \tilde{\sigma}_i
\end{align*}
for the subspaces~$S_{I_s} = \operatorname{span}\{u_{ij} : j\in I_s\}$.
Assembled, this yields~\eqref{eq:bound:approx_Prop4.1}.
\end{proof}

\begin{remark}[Clustering]\label{rem:clustering}
In many applications, data is naturally available as triplets of the form
\[
(\Psi(\bar{x}_p), \bar{u}_p, \nabla\Psi(\bar{x}_p)^\top f(\bar{x}_p,\bar{u}_p))
\in \mathbb{X}\times\mathbb{U}\times\mathbb{R}^M,\quad p\in\mathbb{N}.
\]
To obtain a data set of the form~\eqref{eq:data:clustered} used in this work, 
one may choose a finite set of representative states 
$\mathcal{X}=\{x_1,\dots,x_d\}\subset\mathbb{X}$ 
and cluster the available samples by proximity in the observable space, 
assigning $(\bar{x}_p,\bar{u}_p)$ to the cluster of $x_i$ if 
$\|\Psi(\bar{x}_p)-\Psi(x_i)\|\le r_{x_i}$. 
Indexing samples in the $i$th cluster by $j$ yields data points
\[
(\Psi(x_{ij}),u_{ij},\nabla\Psi(x_{ij})^\top f(x_{ij},u_{ij})) 
\in \mathcal{B}_{L_\Psi r_{x_i}}(\Psi(x_i))\times\mathbb{U}\times\mathbb{R}^M,
\]
which provides the structured data set required for the analysis.
\end{remark}

\begin{remark}[EDMD for discrete-time systems]\label{rem:EDMD:operator}
    We consider the discrete-time control-affine system
\begin{align}\label{eq:system:DT}
    x^+ = F(x, u) = g_0(x) + G(x)u = g_0(x) + \sum\limits_{k = 1}^m {g}_k(x) u_k
\end{align}
with nonlinear locally Lipschitz maps $g_0: \mathbb{X} \rightarrow \mathbb{R}^n$ and $G:\mathbb X\rightarrow \mathbb R^{n\times m}$. 
Such systems are often derived from continuous-time systems \eqref{eq:sys:continuous} by discretization. Using a Taylor expansion, we obtain a discrete-time system up to an error of order~$\mathcal{O}(\Delta t^2)$, see \cite[Remark 4.1]{BoldPhil24}. \\
Analogously to the generator setting, using the excitation-based approach from Section~\ref{sec:excitation}, the components 
$g_0(x_i), G(x_i)$ for $x_i \in \mathcal{X}$ for all $i \in \{1, \dots, d\}$ can be approximated from data of the form 
\begin{align}\label{eq:data:operator}
    (x_{ij},u_{ij},F(x_{ij},u_{ij})) \in  \mathcal{B}_{r_{x_i}}(x_i)
    \times \mathbb{U} \times \mathbb{R}^n
\end{align} 
for $i\in [1:d]$ and $j \in [1:d_i]$ with $d_i \geq m + 1$, $r_{x_i} \geq 0$ and a set of points~$\mathcal{X} = \{x_1, \dots, x_d\} \subseteq \mathbb{X}$. \\
This enables the construction of an artificial sample set~
$(x_i,F(x_i,e_k))$ for $e_k$ being the $k$-th unit vector of $R^m$, $k \in \{1, \dots, m\}$, and $e_0 = 0$ using~\eqref{eq:system:DT}. 
Now, for an observable vector $\Psi = (\psi_1, \dots, \psi_M)^\top$ and $k\in\{0,1,\dots,m\}$, the matrices 
\begin{align*}
    X = \begin{bmatrix}
        \Psi(x_1) & \dots & \Psi(x_d) \\
    \end{bmatrix} \quad \text{and} \quad 
    Y^k = \begin{bmatrix}
        \Psi(F^\varepsilon(x_1, e_k)) & \dots & \Psi(F^\varepsilon(x_d, e_k))
    \end{bmatrix}
\end{align*}
 are assembled, where $F^\varepsilon(x_i, e_k)$ stands for the approximated values of $F(x_i, e_k)$. Then an approximation of the compressed Koopman operator~$P_{\mathbb{V}}\mathcal{K}^i|_{\mathbb{V}}$, $i\in\{0,1,\dots,m\}$, on the space~$\mathbb{V} = \operatorname{span}\{\psi_p \mid p \in \{1, \dots ,M\}\}$ is given by the solution of the regression problem 
\begin{align} \label{eq:operator:regression}
    K^i = \argmin_{K \in \R^{M \times M}} \|KX - Y^i\|_F^2.
\end{align}
    
\end{remark}

\subsection{Kernel EDMD with flexible sampling} \label{subsec:kernel:EDMD}
Kernel EDMD~\cite{klus2020kernel} yields an approximation of the Koopman operator based on a data-dependent dictionary, where only the kernel has to be chosen.
In this subsection, we recall the results from \cite[Section~4]{BoldPhil24} where an extension of kernel EDMD (kEDMD) to control-affine systems accompanied by bounds on the full approximation error was introduced.

Let the function~$\k:\R^n \times \R^n \rightarrow \R$ be a symmetric and strictly positive definite kernel function, i.e., for all sets of states $\mathcal{X} = \{x_1, \dots, x_d\} \subset \R^n$, the corresponding kernel matrix~$K_\mathcal{X} = \left(\k(x_i, x_j)\right)_{i, j = 1}^d$ is positive definite. For $x \in \R^n$, we define the \textit{canonical features} of $\k$ as ~$\Phi_x:\R^n \rightarrow \R$ with $y \mapsto \k(x, y)$ for all $y \in \R^n$.
The completion of $\operatorname{span}\{\Phi_x \mid x \in \R^n\}$ yields a Hilbert space~$\mathbb{H}$ induced by $\k$ with inner product~$\langle \cdot, \cdot\rangle_\mathbb{H}$. The characteristic of~$\mathbb{H}$ includes that its elements~$f$ fulfill 
\begin{equation}\label{eq:RKHS:reproducing-property}\tag{reproducing property}
    f(x) = \langle f, \k(x, \cdot) \rangle_\mathbb{H} = \langle f, \Phi_x(\cdot) \rangle_\mathbb{H} \qquad\forall\, x \in \mathbb{R}^n.
\end{equation}
Hence, $\mathbb{H}$ is referred to as \textit{reproducing kernel Hilbert space} (RKHS).\\

\begin{remark}[Wendland Kernels] \label{rem:Wendland}
    A suitable choice of kernel functions is based on the Wendland radial basis functions~$\Phi_{n, k}:\R^n \rightarrow \R$ with smoothness degree~$k \in \N$ from \cite{wendland2004approximate}. These functions induce a piecewise-polynomial and compactly-supported kernel function with 
    \begin{align*}
        \k(x, y) \coloneqq \Phi_{n, k}(x - y) = \phi_{n, k}(\|x - y\|) 
    \end{align*}
    for $x, y \in \R^n$ and $\phi_{n, k}$ defined as in \cite[Table 9.1]{wendland2004approximate}.
    Note that the corresponding RKHS induced by the Wendland kernels on a bounded domain~$\Omega \subset \R^n$ with Lipschitz boundary coincides with a fractional Sobolev space. 
\end{remark}

\textit{Kernel EDMD} embedded in a suitable RKHS~$\mathbb{H}$ leverages kernel methods to define a high-dimensional data-informed feature space on which the Koopman operator can be approximated, see, e.g., \cite{will2015kernel,klus2020kernel}. 
Consider an autonomous discrete-time system given by 
\begin{align}\label{eq:sys:ADT}
    x^+ = F(x)
\end{align}
with locally Lipschitz continuous map~$F:\mathbb{X} \rightarrow \R^n$ and where $\mathbb{X} \subset \R^n$ is an open and bounded set with Lipschitz boundary containing the origin in its interior. In the following, we assume that the RKHS~$\mathbb{H}$ induced by the kernel function~$\k$ is invariant w.r.t.\ the Koopman operator, i.e., $\mathcal{K}\mathbb{H} \subseteq \mathbb{H}$. An example of kernel functions fulfilling this property are Wendland kernels (see Remark~\ref{rem:Wendland}).\\
For a set of pairwise distinct data points~$\mathcal{X} = \{x_1, \dots, x_d\}\subset\R^n$, $d \in \N$, we define the $d$-dimensional set~$V_\mathcal{X} = \operatorname{span}\{\Phi_{x_1}, \dots, \Phi_{x_d}\}$. As proven in \cite[Proposition 3.2]{kohne2025error}, a kEDMD approximant~$\widehat{K}$ of the Koopman operator~$\mathcal{K}$ is given by the orthogonal projection~$P_{V_\mathcal{X}}$ of $\mathbb{H}$ onto $V_\mathcal{X}$, i.e., the compression of the  Koopman operator~$P_\mathcal{X}\mathcal{K}|_{V_{\mathcal{X}}}$, 
with matrix approximant 
\begin{align*}
    \widehat{K} = K_\mathcal{X}^{-1}K_{F(\mathcal{X})}K_\mathcal{X}^{-1} \in \R^{d \times d},
\end{align*}
see \cite{kohne2025error} with $K_{F(\mathcal{X})} = (\k(F(x_i), x_j))_{i, j = 1}^d \in \R^{d \times d}$. 
For 
$\mathbf{k}_\mathcal{X} = (\Phi_{x_1}, \dots, \Phi_{x_d})^\top$ and $\psi_{F(\mathcal{X})} = (\psi(F(x_1)), \dots, \psi(F(x_d)))^\top$, the data-driven surrogate dynamics~$F^\varepsilon$ is then given by 
\begin{align}\label{eq:surr}
    \psi(F(x)) \approx \psi(F^\varepsilon(x)) \coloneqq \psi_\mathcal{X}^\top \widehat{K}^\top\mathbf{k}_\mathcal{X}(x)
\end{align}
with $\psi:\R^n \rightarrow \R$.
Hereby, $\varepsilon$ in the notation of the surrogate right-hand side stands for the approximation quality of the kEDMD-based model. In~\cite[Theorem~5.2]{kohne2025error}, a bound on the full approximation error is derived, where the approximation accuracy of the kEDMD-based surrogate depends on the fill distance
\begin{align*}
    h_\mathcal{X} := \sup_{x \in \mathbb{X}} \min_{x_i \in \mathcal{X}} \|x - x_i \|
\end{align*}
of $\mathcal{X}$ in $\mathbb{X}$. The fill distance, measuring how well the set of target points covers the space~$\mathbb{X}$, is a common measure in machine learning and approximation theory and  plays a major role when considering error bounds. Theorem~\ref{thm:koehne} recalls said result. 
\begin{theorem}
    \label{thm:koehne}
    Let $\mathbb{H}$ be the RKHS on~$\mathbb{X}$ generated by the Wendland kernels with smoothness degree $k \in \N$. 
    The right-hand side of system~\eqref{eq:sys:ADT} shall be given by ${F \in \mathcal{C}^p_b(\mathbb{X},\mathbb{R}^n)}$ , with $p = \lceil \frac{n + 1}{2} + k\rceil$, where $C_b^p$ denotes the space of bounded $p$-times continuously differentiable functions. 
    Then, there exist constants $C,h_0 > 0$ such that the bound on the full approximation error 
    \begin{align}\label{eq:error_bound}
        \| \mathcal{K} - \widehat{K}\|_{\mathbb{H} \rightarrow \mathcal{C}_b({\mathbb{X}},\R^n)} \le Ch_\mathcal{X}^{k+1 \mathbin{/} 2}
    \end{align}
    holds for all sets $\mathcal{X} :=\{ x_i \mid i \in \{1, \dots, d\} \} \subset \mathbb{X}$, $d \in \mathbb{N}$, of 
    pairwise-distinct data points with fill distance~$h_{\mathcal{X}}$, $h_{\mathcal{X}} \leq h_0$. 
\end{theorem}

\textbf{Extension to control-affine systems}.
For discrete-time control-affine systems given by~\cref{eq:system:DT}, \cite[Section 4]{BoldPhil24} provide a data-driven kernel EDMD extension. 
The main idea is to use the autonomous kernel EDMD method to approximate the functions~$g_0, \dots, g_m$ and then insert them into a control affine form such as \eqref{eq:system:DT}. 
To realize this approximation, data points of the form~$(x_i, g_k(x_i))$ are needed, where $k \in \{0, \dots, m\}$ and $i \in \{1, \dots, d\}$. Analogously to bilinear EDMDc in Remark~\ref{rem:EDMD:operator}, under the assumption that we are given data points \begin{align}\label{eq:data:kEDMD}
(x_{ij},u_{ij},F(x_{ij},u_{ij})) \in  B_{r_{x_i}}(x_i) \times \mathbb{U} \times \mathbb{R}^n
\end{align} 
for $i\in \{1, \dots, d\}$, $j \in \{0, \dots, d_i\}$ with $d_i \geq m$, and radii~$r_{x_i} \geq 0$, we can compute approximations~$\tilde{g}_k(x_i)$ of the data points~$g_k(x_i)$ for $k \in \{ 0, \dots, m\}$. \\ 
The coefficients for the data-driven surrogate~$F^\varepsilon$ are then computed analogously to the autonomous case~\eqref{eq:surr}, which leads to the following propagation step
\begin{align*}
    \psi(F(x, u)) \approx \psi(F^\varepsilon(x, u)) \coloneqq \psi_\mathcal{X}^\top \Big(\widehat{K}_0 + \sum_{k = 1}^m \widehat{K}_k u_k\Big)^\top \mathbf{k}_\mathcal{X}(x)
\end{align*}
with $\widehat{K}_k = K_\mathcal{X}^{-1}K_{\tilde{g}_k(\mathcal{X})}K_\mathcal{X}^{-1}$, $K_{\tilde{g}_k(\mathcal{X})} = (\k(\tilde{g}_k(x_i), x_j))_{i, j = 1}^{d}$, for all $k \in \{0, \dots, m\}$. 
We then obtain the full state-space surrogate model 
\begin{align*}
    x^+ = F^\varepsilon(x, u)
\end{align*}
by using the observable functions~$\psi_\ell(x) = e_\ell^\top x$, i.e., the $\ell$th coordinate function for $\ell \in \{1, \dots, n\}$.\\
The uniform error bound on the full approximation error in Theorem~\ref{thm:control_main} extends the results from~\cite[Thm.\ 3]{schimperna2025data} including our findings from Theorem~\ref{t:main} in  Section~\ref{sec:excitation}.

\begin{theorem}[Approximation error]\label{thm:control_main}
  Let $k \geq 1$ be the smoothness degree of the Wendland kernel. 
  Further, let $\mathcal{X} = \{x_1, \dots, x_d\} \subset \R^n$ and data according to~\eqref{eq:data:kEDMD} be given such that $V_i$ defined by \eqref{eq:g:V} 
    has full row rank, i.e., $\operatorname{rank}(V_i) = m + 1$ for $i \in \{1, \dots, d\}$.
    Moreover, let the control inputs be arranged such that $U_m = \big[u_{i1} \mid \cdots \mid u_{im}\big]$ is invertible and let $\{\iota_1,\ldots,\iota_m\} = \{1, \dots, m\}$ such that $\|u_{\iota_1}\|\ge\cdots\ge\|u_{\iota_m}\|$, and set $I_j = \{1, \dots, m\}\backslash\{\iota_1,\ldots,\iota_j\} = \{\iota_p : p>j\}$.       
    Based on this, the subspaces~$S_{I_s}$ are defined by $ S_{I_s}= \operatorname{span}\{u_{ij} : j\in I_s\}$.
    Then, there exist constants $C_1, C_2, h_0 > 0$ such that the error bound 
\begin{equation}\label{eq:kEDMD:error bound}
        \|F(x,u) - F^\varepsilon(x,u)\|_\infty \leq C_1 h_\mathcal{X}^{k+1/2} + C_2 c \|K_\mathcal{X}^{-1}\|r_\mathcal{X}\tilde{\sigma}
    \end{equation}
    holds for all $(x,u) \in \mathbb{X} \times \mathbb{U}$, where the data is given by \eqref{eq:data:kEDMD} with fill distance~$h_\mathcal{X}$ and cluster radius~$r_{\mathcal{X}} \coloneqq \max\{r_{x_i} \mid i \in \{1, \dots, d\}\}$ satisfying $h_\mathcal{X} \leq h_0$ and $r_{\mathcal{X}} < h_{\mathcal{X}}/2$, respectively. Thereby, the error bound depends on the cluster size $r_\mathcal{X}$, on the fill distance $h_\mathcal{X}$, on $\tilde{\sigma} > 0$ with 
    \begin{align*}
        \tilde{\sigma}_i \coloneqq \Big(\Theta(U_m^{-1}u_{i0}) \cdot \min\{m + 1, \|u_{i\iota_m}\| \prod_{s = 1}^{m - 1}(1 - \cos\theta(u_{i\iota_s}, S_{I_s}))\}\Big)^{-\frac12},
    \end{align*} 
    and on a constant~$c$ defined as $c := \Phi_{n,k}^{1/2}(0)\big( \max_{v \in \mathbb{R}^d: \| v \|_\infty \leq 1} v^\top {K}_\mathcal{X}^{-1}v \big)^{1/2}$.
\end{theorem}
\begin{proof}
Following the proof of \cite[Thm. 3]{schimperna2025data}, we get an estimation of the approximation error by 
\begin{align*}
    \|F(x,u) - F^\varepsilon(x,u)\|_\infty \leq C_1 h_\mathcal{X}^{k+1/2} + C_2 c \|K_\mathcal{X}^{-1}\|r_\mathcal{X} \cdot \max_{i \in [1:d]}\|V_i^\dagger\|
\end{align*}
for constants $C_1, C_2 > 0$ and where $c \coloneqq \Phi_{n,k}^{1/2}(0)\big( \max_{v \in \mathbb{R}^d: \| v \|_\infty \leq 1} v^\top {K}_\mathcal{X}^{-1}v \big)^{1/2}$. \\ 
The estimation of the term~$\|V_i^\dagger\|$ can be done analogously to the proof of Proposition~\ref{cor:approx:error:generator}. 
Due to the full row rank, we can write $V_i = \Big[\tilde{V}_i \mid \bar{V}_i\Big]$ with $\tilde{V}_i = \begin{bmatrix}\mathbbm{1}_{m+1}^\top\\U_{m+1}\end{bmatrix} \in \R^{m + 1\times m + 1}$ with $U_{m+1} = \big[ {u}_{i0} \mid \cdots \mid {u}_{im} \big]$, such that $U_{m} = \big[ {u}_{i1} \mid \cdots \mid {u}_{im} \big]$ is invertible. For the permutation set~$I_j$ and with Theorem~\ref{t:main} and Remark~\ref{rem:extending_V}, we obtain 
\begin{align*}
   \sigma_{\min} (V_i)^{-1} &\leq \sigma_\mathrm{min}(\tilde{V}_i)^{-1} \\
    &\leq \Big(\Theta(U_m^{-1}u_{i0}) \cdot \min\{m + 1, \|u_{i\iota_m}\| \prod_{s = 1}^{m - 1}(1 - \cos\theta(u_{i\iota_s}, S_{I_s}))\}\Big)^{-\frac12}\eqqcolon \tilde{\sigma}_i
\end{align*}
for the subspaces~$S_{I_s} = \operatorname{span}\{u_{ij} : j\in I_s\}$. This gives us the error bound in \eqref{eq:kEDMD:error bound}.
\end{proof}

\begin{remark}[Kernel generator EDMD]
Instead of using kernel EDMD to learn a surrogate model for discrete-time dynamics, the same method can be used to approximate a continuous-time system of the form \eqref{eq:sys:continuous}. Therefore, we consider data given by \eqref{eq:data:clustered}
and use the same method as in Section~\ref{subsec:generator:EDMD} to approximate the data set 
\begin{align}
    (x_i, g_0(x_i), G(x_i)) \quad \text{ for } \quad i \in \{1,\ldots,d\}
\end{align}
with $\Psi:\R^n \rightarrow \R^n$ as the identity function. Using these artificial data points, we can use kernel EDMD as described above to approximate the functions $g_0$ and~$G$.
\end{remark}

    \section{Applications}
\label{sec:application}
To illustrate the findings, we investigate several excitation strategies in two application scenarios.
First, a six-dimensional input space~($m=6$) is considered.
Besides analyzing the minimum singular values of the input-data matrices~$V$, we study how the choice of input affects the approximation error of the vector fields associated with the approximately discretized kinematics of a free-floating rigid body.
In the second example, we investigate the excitation strategies for data collection of a nonholonomic mobile robot.
In the latter setting, inspired by practical robotics, the complete pipeline is investigated, including the use of the identified discrete-time vector fields to generate artificial data for bilinear EDMDc for control-affine systems.

For both application examples, we consider the same four excitation strategies, independent of the input dimension~$m$ and the specific control-affine mapping to be identified.
The colors introduced in the following are used consistently throughout all subsequent illustrations. 
The first strategy $U_\textnormal{r}$~(blue) randomly draws inputs $u\in\mathbb{U}$ according to a uniform distribution.
The second strategy~$U_\perp$~(red) selects the inputs according to Cor.~\ref{cor:scaling}, i.e., it uses an orthogonal basis of $\R^m$ and an additional input such that~\eqref{eq:NCO} from Lem.~\ref{lem:NCO} is satisfied.
The third strategy~$U_\triangle$~(orange) is based on Prop.~\ref{prop:simplex}, i.e., it selects $m+1$ inputs corresponding to the vertices of a regular $m$-simplex. 
The strategies~$U_\perp$ and~$U_\triangle$ can be seen as offline strategies, since they assume full control over all $\dd+1\geq m+1$ applied inputs.
In practice, however, such freedom may not always be available, particularly in a local neighborhood of a given observation point in the state space.
Motivated by this setting, the fourth strategy~$U_{\measuredangle}$~(pink) deliberately constructs the $(m+1)$th input when $m$ random inputs are given.
For the high-dimensional example, the subspace angles of the random inputs are evaluated, and the input associated with the smallest subspace angle is deliberately replaced according to Thm.~\ref{t:main} in order to maximize the minimum singular value of the resulting input-data matrix.
In the low-dimensional example, the $(m+1)$th input is constructed analogously once the first $m$ random inputs are given, i.e., the input is appended rather than necessarily replacing the random input associated with the smallest subspace angle.
For comparability, the random inputs used in both scenarios coincide with the inputs of $U_{\textnormal{r}}$. 

\subsection{Kinematics of a Free-Floating Rigid Body}\label{sub:airship}
We first investigate the four excitation strategies~$U_\textnormal{s}$, $\textnormal{s}\in\{\textnormal{r},\perp, \triangle,\measuredangle\}$, for a six-dimensional input space, where a total of $d=10^4$ observation points is considered.
For the random strategy~$U_{\textnormal{r}}$, $m+1=\dd+1=7$ inputs~$u_{ij}\in\mathbb{U}$, $j\in[0:\dd]$, are sampled independently for each observation point~$i\in[d]$, where~$\mathbb{U}=\left\lbrace u\in\mathbb{R}^m  : \| u \| \leq r_u \right\rbrace$ with $r_u=10$.
For the online strategy~$U_{\measuredangle}$, at each observation point~$i$, the inputs~$u_{ij}$ follow based on the random strategy~$U_{\textnormal{r}}$ as described above.
The inputs of the second and third strategy are scaled by~$\alpha=\sqrt{\dd+1}$ in order to satisfy the conditions of Cor.~\ref{cor:scaling} and Prop.~\ref{prop:simplex}, respectively.
Applying the four input strategies to each point~$i\in[d]$, Fig.~\ref{fig:airship}~(left) shows the resulting minimum singular values of the corresponding input-data matrices $V_{i,\textnormal{s}} \coloneqq V_i(U_\textnormal{s})$, $\textnormal{s}\in\{\textnormal{r},\perp, \triangle,\measuredangle\}$, sorted in descending order w.r.t.~$i$.
As derived in Cor.~\ref{cor:scaling} and Prop.~\ref{prop:simplex}, the strategies~$U_\perp$ and~$U_\triangle$ attain the maximum minimum singular value~$\sqrt{\dd+1}$.
For visual clarity, Fig.~\ref{fig:airship}, as well as all subsequent plots of this type, only displays a subset of the corresponding points.
In contrast, the random strategy~$U_\textnormal{r}$ yields substantially smaller minimum singular values.
Importantly, the proposed strategy~$U_{\measuredangle}$ achieves a significant improvement by deliberately selecting the $(\dd+1)$th input.
In many cases, it approximately meets the value of the free-choice strategies~$U_\perp$ and~$U_\triangle$, see Fig.~\ref{fig:robot_singular_values}~(left).
\begin{figure*}
	\centering
    \includegraphics[scale=0.8]{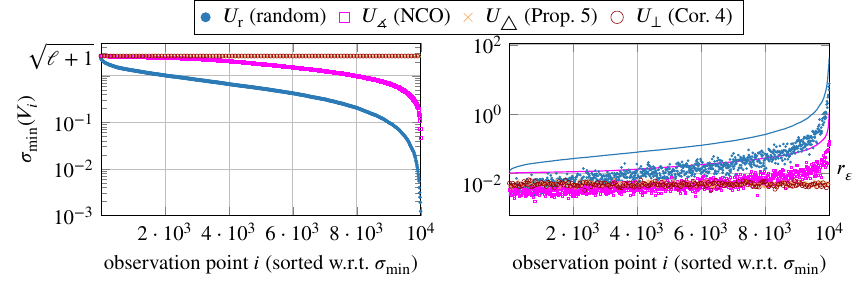}
    \caption{
    Minimum singular values of~$V$ for the different input strategies~$U_{\textnormal{s}}$ ($\dd+1=m+1=7$) sorted in descending order~(left) and the corresponding approximation errors $\left\lVert \begin{bmatrix}
            g_0(x_i) & G(x_i)
        \end{bmatrix} - \begin{bmatrix}
            \hat{g}_0^\star & \widehat{G}^\star
        \end{bmatrix}\right\rVert_\mathrm{max}$ of the approximately discretized vector fields~\eqref{eq:kinematics_airship} at the observation points~$x_i$~(right), see also Prop.~\ref{prop:approx:error}.
        }
    \label{fig:airship}
\end{figure*}

Crucially, as shown in Prop.~\ref{prop:approx:error}, the minimum singular value of~$V$ affects the approximation quality of control-affine mappings, e.g., when identifying the vector fields at some observation point~$x_i$ for discrete-time control systems.
To this end, we consider the approximately discretized kinematics of a free-floating rigid body in the space, which can be expressed by
\begin{align}\label{eq:kinematics_airship}
    x^+ = 
    \begin{bmatrix}
        p^+ \\ \Theta^+
    \end{bmatrix} =
    x + \delta_t \begin{bmatrix}
        R (\Theta) & 0 \\
        0 & T (\Theta) 
    \end{bmatrix}
    \begin{bmatrix}
        v \\ \omega
    \end{bmatrix}
    = g_0 (x) + G(x) u 
\end{align}
with~$p, \, \Theta \in\mathbb{R}^3$ denoting the body's position and orientation (using Euler angles), respectively, and $v,\, \omega \in\mathbb{R}^3$ denoting the translational velocity and rotational velocity, respectively, in the body-fixed frame.
Moreover, $R,T\!: \mathbb{R}^3 \to \mathbb{R}^{3\times 3}$ is the rotation matrix and angular-velocity transformation matrix, respectively.
The kinematics~\eqref{eq:kinematics_airship} is of relevance, for example, when considering fully actuated airships or underwater vehicles, see, e.g.,~\cite{ChaabaniAzouz22,DueckerEtAl20}, particularly when employing an underlying low-level controller governing the body velocities such that the state and input follow as $x=[p^\top \ \Theta^\top]^\top\in\mathbb{X}\subseteq \mathbb{R}^n$, $n=6$, and $u = [v^\top \ \omega^\top]^\top \in\mathbb{U}\subset \mathbb{R}^m$, $m=6$, respectively.
In the following, the $d$ observation points~$x_i$ are sampled uniformly from the cube~$[-1, \, 1]^3\, \unit{m}$ for the position coordinates and $[-10, \, 10]^3 \,\unit{^\circ}$ for the Euler angles, thereby avoiding singularities in~$T(\Theta)$.
The discretization time is chosen as~$\delta_t = \unit[10]{ms}$.
The actual samples~$x_{ij}$, at which the corresponding inputs are applied, are generated via rejection sampling from a uniform distribution over~$\mathcal{B}_{r_x}(x_i)$ with~$r_x=10^{-2}$.

Using the four excitation strategies, successor states~$x_{ij}^+$, $i\in[d]$, $j\in[0:\dd]$, are generated to approximate the vector fields~$g_0$ and $G$ at the observation points~$x_i$, cf.~Rem.~\ref{rem:EDMD:operator}.
The same randomly sampled observation points~$x_i$ and neighboring states~$x_{ij}$ are utilized for all four strategies.
Figure~\ref{fig:airship}~(right) illustrates the resulting approximation errors of the vector fields for the different excitation strategies~$U_{\textnormal{s}}$, evaluated according to Prop.~\ref{prop:approx:error}.
The strategies~$U_\perp$ and~$U_\triangle$ yield approximation errors of comparable magnitude, consistent with the theoretical upper bound~$r_\varepsilon$, indicated in Fig.~\ref{fig:airship}~(right).
For the online strategy~$U_{\measuredangle}$, the approximation error remains of similar magnitude for most observation points~$i$, with the corresponding upper bound from Prop.~\ref{prop:approx:error} shown as the pink solid line.
In contrast, the random strategy $U_\textnormal{r}$ produces both a larger theoretical upper bound (solid blue line), due to smaller values of~$\sigma_{\min}$, and noticeably larger empirical errors (blue markers).
Overall, Fig.~\ref{fig:airship}~(right) demonstrates that the choice of excitation inputs substantially affects the accuracy of the approximated vector fields~$\hat{g}_0^\star$ and~$\hat{G}^\star$ at $x_i$.
Moreover, deliberately selecting even a single input based on previously sampled random inputs can significantly improve the approximation quality.
The following example investigates how this further impacts bilinear EDMDc for learning control-affine systems, as an example of how to use the findings in a robotics scenario.

\subsection{Nonholonomic Mobile Robot}
In addition to analyzing the excitation properties of different input strategies for a nonholonomic mobile robot, we investigate how these strategies affect bilinear EDMDc with flexible sampling for approximating the Koopman operator, see Sec.~\ref{sec:Koopman}.
In particular, we consider an approximately discretized, control-affine version of the robot kinematics, given by
\begin{align}\label{eq:kinematics_robot}
    x^+ = g_0 (x) + G(x) u = x + \delta_t \begin{bmatrix} \frac{R}{2} \cos x_3 & \frac{R}{2} \cos x_3 \\[3pt] \frac{R}{2} \sin x_3 & \frac{R}{2} \sin x_3 \\[3pt] -\frac{R}{L} & \frac{R}{L}    \end{bmatrix} \begin{bmatrix} \dot{\varphi}_{\ell} \\ \dot{\varphi}_{\textnormal{r}} \end{bmatrix},
\end{align}
where~$\delta_t\in\mathbb{R}_{>0}$ denotes the sampling time.
The state $x\in\mathbb{X} \subseteq \mathbb{R}^n$, $n=3$, 
consists of the robot's planar position $\begin{bmatrix}x_1 & x_2\end{bmatrix}^\top$ and its heading angle~$x_3$ w.r.t.\ the positive $x$-axis of the inertial frame, see, e.g.,~\cite{rose25Koopman} for a schematic illustration.
The robot is actuated through the angular velocities of the left and right wheels, respectively, resulting in the input $u=\begin{bmatrix} \dot{\varphi}_\ell & \dot{\varphi}_{\textnormal{r}}\end{bmatrix}^\top \in\mathbb{U}\subset \mathbb{R}^m$, $m=2$, where $\mathbb{U}=\left\lbrace u\in\mathbb{R}^m  : \| u \| \leq \unitfrac[20]{rad}{s} \right\rbrace$.
The driven wheels have radius~$R$ and are mounted at a distance~$L$ along a common axis.
Note that we deliberately consider only the approximate discrete-time formulation~\eqref{eq:kinematics_robot} of the real-world vehicle to isolate, as far as possible, the errors induced by the excitation strategy itself.
This avoids additional errors arising from the exact discretization, which is generally not control-affine for $\delta_t > 0$, compare~\cite{worth2015nonhol}.
Nevertheless, for $\delta_t \to 0$,~\eqref{eq:kinematics_robot} coincides with the exact discretization.

In the following, $d=180$ observation points are drawn from a uniform distribution within the plane~$[-0.5, \, 0.5]^2 \, \unit{m}$, whereas the samples are spaced in an equidistant fashion in the orientation~$x_3\in[0, \, 2\pi)$.
The actual samples~$x_{ij}$ required to construct the artificial data set are again generated via rejection sampling from a uniform distribution over~$\mathcal{B}_{r_x}(x_i)$ with~$r_x=10^{-3}$.
In addition to considering how the choice of applied inputs~$u_{ij}$, $j\in[0:\dd]$, affects the intermediate approximations~$\hat{g}_0^\star$ and~$\hat{G}^\star$ at $x_i$, see also the previous example, it is essential to investigate how these choices ultimately influence the accuracy of the derived Koopman-based surrogate models, since the approximated vector fields are required for flexible sampling in bilinear EDMDc.
Importantly, for the surrogate models, also other characteristics, such as the finite-dimensional dictionary, play a significant role.
\begin{figure*}
    \centering
	\includegraphics[scale=0.8]{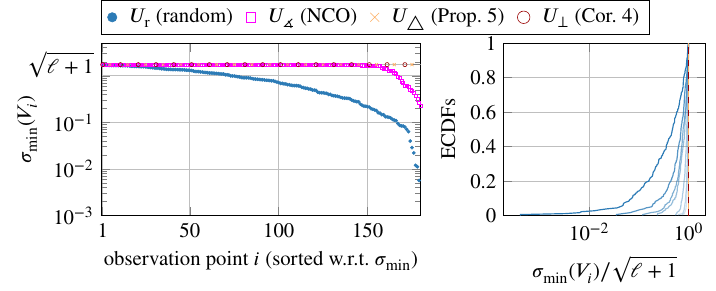}
    \caption{Minimum singular values over the observation points, sorted w.r.t.\ $\sigma_{\min}$, for $\dd+1=m+1=3$ neighbors~(left) and the empirical cumulative distribution functions (ECDFs) of the normalized minimum singular values for~$U_{\textnormal{r}}$, $U_\triangle$, and $U_\perp$ with varying numbers~$(\dd+1)\in\{3,\, 4,\, 5,\, 6,\, 10,\, 20,\, 30\}$, where the line opacity decreases with increasing~$\dd$~(right).
    }
    \label{fig:robot_singular_values}
\end{figure*}

Again, the aforementioned four input strategies~$U_\textnormal{s}$, $\textnormal{s}\in\{\textnormal{r},\perp,\triangle,\measuredangle\}$ are employed.
The inputs of the second and third strategies are scaled by~$\alpha=2\pi$, ensuring that $\alpha \geq \sqrt{\dd+1}$. 
Figure~\ref{fig:robot_singular_values}~(left) shows the resulting minimum singular values of the corresponding input-data matrices $V_{i,\textnormal{s}} \coloneqq V_i(U_\textnormal{s})$, $\textnormal{s}\in\{\textnormal{r},\perp, \triangle,\measuredangle\}$ for the case~$\dd+1=m+1=3$ for all $i$, sorted in descending order.
The results are consistent with the observations from the previous higher-dimensional example, see Sec.~\ref{sub:airship} for a detailed discussion.
Since it is sufficient to have $m+1$ directions spanning large subspace angles, see also Rem.~\ref{rem:extending_V}, the probability that the randomly sampled inputs yield larger values of~$\sigma_{\min}(V_{i, \textnormal{r}})$ naturally increases with~$\dd$.
Consequently, the distribution of~$\sigma_{\min}(V_{i,\textnormal{r}})$ becomes more concentrated and approaches the bound~$\sqrt{\dd+1}$ as~$\dd$ increases.
This effect is illustrated in Fig.~\ref{fig:robot_singular_values}~(right), which shows the empirical cumulative distribution functions~(ECDFs) of the normalized minimum singular values for different input strategies and increasing numbers of inputs~$\dd+1$, where the opacity decreases with increasing~$\dd$.
The approximation errors of the vector fields at the observation points exhibit behaviors qualitatively similar to that observed in the previous higher-dimensional example for the four input strategies~$U_{\textnormal{s}}$.
A detailed discussion of these errors is therefore omitted for brevity and we focus on bilinear EDMDc instead. 
To this end, the approximated input vector fields are used to generate the artificial data points~${F}^\varepsilon(x_i, e_k)$, compare Rem.~\ref{rem:EDMD:operator}.
These data points are then lifted to construct the matrices required for bilinear EDMDc of control-affine systems, see Sec.~\ref{sec:Koopman}.
As dictionary, we choose $\Psi=\{ 1,\,x_1,\,x_2,\,\cos x_3,\,\sin x_3 \}$ since this choice has shown good results in previous studies~\cite{bold2024robot, rose25Koopman}.
The orientation is reprojected using the four-quadrant inverse tangent. 
This reprojection does introduce additional errors, meaning that the following results are not solely influenced by the excitation strategy used for flexible sampling.
For a given reference input trajectory that nominally yields a lemniscate-shaped trajectory, Fig.~\ref{fig:robot_lemniscate} shows the resulting open-loop trajectories~(left) of the Koopman surrogate models obtained from flexible sampling with different excitation strategies, together with the corresponding Euclidean and orientation one-step errors~(right). 
As illustrated, for $\dd + 1=m+1=3$, the surrogate model derived from randomly chosen inputs~$U_{\textnormal{r}}$ exhibits significantly larger one-step errors compared to the other three strategies.
Importantly, choosing the $(m+1)$th input according to Thm.~\ref{t:main} already appears to substantially improve the accuracy of the resulting Koopman model.
However, increasing the number of random inputs for~$U_{\textnormal{r}}$ to $\dd +1= 4$ also considerably improves the approximation quality, cf.\ Fig.~\ref{fig:robot_singular_values}~(right), which is likewise reflected in Fig.~\ref{fig:robot_lemniscate} (blue dashed line).
\begin{figure*}
    \centering
	\includegraphics[scale=0.8]{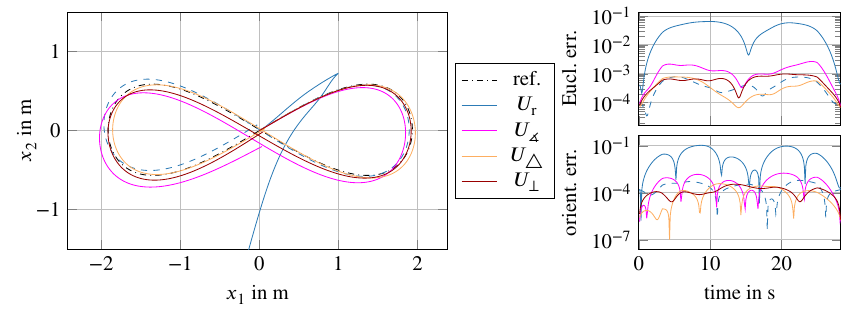}
    \caption{Open-loop trajectories~(left) of the Koopman surrogate models obtained with different excitation strategies during flexible sampling, along with the corresponding one-step prediction errors~(right).
	The solid lines correspond to $\dd+1=m+1=3$ neighbors, while the dashed blue for~$U_{\textnormal{r}}$ additionally shows the case~$\dd+1=4$.}
    \label{fig:robot_lemniscate}
\end{figure*}

    \section{Conclusions}
    In this article, we presented a data-fitting framework for identifying control-affine systems. 
We focused on the role of input-data design in ensuring the quality of the resulting model. 
Our main contribution is the derivation of lower bounds for the minimal singular value of the input-data matrix, which directly impacts the robustness of the respective regression. 
We proposed geometrically interpretable excitation schemes to construct control inputs from scratch, which achieve optimality, and presented a guideline for extending a given set of inputs using angle-based information in the input space, including concrete strategies to choose exciting inputs. 
We showed how these excitation schemes and singular-value bounds can be integrated into bilinear EDMDc within the Koopman framework, enabling flexible sampling. Furthermore, we used our results to prove uniform error bounds on the full approximation error for kernel EDMD for control-affine systems.
Finally, we exemplified the findings for the different excitation schemes for an example with six control inputs and validated the effectiveness of the excitation schemes in combination with the flexible sampling EDMDc approach at the example of the system identification of a nonholonomic mobile robot, demonstrating that the theoretical findings are relevant and directly applicable to practical problems from engineering.

\bibliographystyle{abbrv}
\bibliography{refs}

\section{Appendix}\label{sec:appendix}

In this appendix, we prove the preparatory Lemma~\ref{t:kaur_ext} and, then, our main result, i.e., Theorem~\ref{t:main}.

\begin{proof}[Proof of Lemma~\ref{t:kaur_ext}]
As already stated in Remark \ref{r:kaur}, \eqref{e:kaur_ext} trivially holds if $u\in\operatorname{ran} Q$. Hence, let us assume that $u\notin\operatorname{ran}Q$. Let $Q = UU^\top$ with $U\in\R^{n\times q}$, where $q = \operatorname{rank} U = \operatorname{rank} Q < n$. Then
\[
P + Q = uu^\top + UU^\top = \begin{bmatrix}u & U\end{bmatrix}\begin{bmatrix}u^\top\\U^\top\end{bmatrix} = AA^\top,
\]
where $A = \begin{bmatrix}u & U\end{bmatrix}\in\R^{n\times (1+q)}$. Note that $\operatorname{rank} A = 1+q$ as $u\notin\operatorname{ran}Q = \operatorname{ran}U$. Hence, $A^\top A\in\R^{(1+q)\times (1+q)}$ is positive definite, and $\lambda_{\min}(P+Q) = \lambda_{\min}(A^\top A)$. We have
\[
A^\top A = \begin{bmatrix}u^\top\\U^\top\end{bmatrix}\begin{bmatrix}u & U\end{bmatrix} = \begin{bmatrix}\|u\|^2 & u^\top U \\ U^\top u & U^\top U\end{bmatrix}.
\]
In what follows, we shall abbreviate $P_U := P_{\operatorname{ran} U}$ denoting the orthogonal projection onto $\operatorname{ran} U$. The Schur complement of this block matrix representation of $A^\top A$ equals
\[
\|u\|^2 - u^\top U(U^\top U)^{-1}U^\top u = \|u\|^2 - u^\top P_{U}u = \|u\|^2  - \|P_{U}u\|^2 =: s.
\]
We may thus write
\[
A^\top A =
\begin{bmatrix}
    1 & -w^\top \\
    0 & I_q
\end{bmatrix}
\begin{bmatrix}
    s & 0\\
    0 & U^\top U
\end{bmatrix}
\begin{bmatrix}
    1 & 0\\
    -w & I_q
\end{bmatrix},
\]
where $w = -(U^\top U)^{-1}U^\top u$. Hence,
\begin{align}\label{e:schur}
(A^\top A)^{-1} = 
\underbrace{\begin{bmatrix}
    1 & 0 \\
    w & I_q
\end{bmatrix}}_{Z_w} \underbrace{\begin{bmatrix}
    1/s & 0\\
    0 & (U^\top U)^{-1}
\end{bmatrix}}_L
\underbrace{\begin{bmatrix}
    1 & w^\top \\
    0 & I_q
\end{bmatrix}}_{Z_w^\top}.
\end{align}
Therefore, we may estimate
\begin{align*}
    \|(A^\top A)^{-1}\|_2
    &= \|Z_w LZ_w^\top\|_2 = \sup\big\{\langle LZ_w^\top x,Z_w^\top x\rangle : \|x\|=1\big\} \\
    &= \sup\big\{\langle Ly,y\rangle : \|Z_w^{-\top}y\|=1\big\} \\
&= \sup\big\{s^{-1}y_1^2 + \langle (U^\top U)^{-1}y_2,y_2\rangle : \|Z_w^{-\top}y\|=1\big\}\\
&= \sup\big\{ay_1^2 + \|By_2\|^2 : \|Z_w^{-\top}y\|=1\big\}
\end{align*}
where $a = s^{-1}$ and $B = (U^\top U)^{-1/2}$. 
Now, observe that $\|Z_w^{-\top}y\|^2_2 = (y_1 - w^\top y_2)^2 + \|y_2\|^2$. Setting $x_1 = y_1 - w^\top y_2$ and $x_2 = y_2$, we obtain
\[
\|(A^\top A)^{-1}\| = \sup\big\{a(x_1 + w^\top x_2)^2 + \|Bx_2\|^2 : x_1^2 + \|x_2\|^2 = 1\big\}.
\]
We further set $b = \|B^2\|$. Note that $b^{-1} = \|(U^\top U)^{-1}\|^{-1} = \lambda_{\min}(U^\top U)$. We have
\begin{align*}
    a(x_1 + w^\top x_2)^2 + \|Bx_2\|^2
    & = a\big[x_1^2 + 2x_1\cdot w^\top x_2 + (w^\top x_2)^2\big] + \|Bx_2\|^2\\
    &\le a\big[x_1^2 + 2|x_1|\|x_2\|\|w\| + \|w\|^2\|x_2\|^2\big] + b\|x_2\|^2.
\end{align*}
Together with $x_1^2 + \lVert x_2\rVert^2 -2|x_1|\|x_2\| = (\lvert x_1\rvert - \lVert x_2\rVert)^2\geq 0$ and $x_1^2 + \lVert x_2\rVert^2 = 1$, one finds $a(x_1 + w^\top x_2)^2 + \|Bx_2\|^2 \leq a\big[x_1^2 + \|w\| + \|w\|^2\|x_2\|^2\big] + b\|x_2\|^2$.
Recall that $w = -(U^\top U)^{-1}U^\top u$. Hence, $Uw = -P_{U}u$ so that
\begin{align*}
\|u\|^2 - \tfrac 1{a} &= \|u\|^2 - s
= \|P_{U}u\|^2 = \|Uw\|^2 \ge\lambda_{\min}(U^\top U)\|w\|^2 = b^{-1}\|w\|^2.
\end{align*}
Hence, $\|w\|^2\le b(\|u\|^2-\frac 1a) = b\|P_Uu\|^2$ so that, with $\beta = \max\{1,\|u\|^2b\}$,
\begin{align*}
    a(x_1 + w^\top x_2)^2 + \|Bx_2\|^2
    &\le a\left[x_1^2 + \sqrt{b}\|P_Uu\| + b(\|u\|^2-\tfrac 1{a})\|x_2\|^2\right] + b\|x_2\|^2\\
    &= a\left[x_1^2 + \sqrt{b\|u\|^2}\frac{\|P_Uu\|}{\|u\|} + b\|u\|^2\|x_2\|^2\right] \\
    &\le a\left[\beta + \sqrt{\beta}\frac{\|P_Uu\|}{\|u\|}\right]\\
&\le a\beta\left[1 + \frac{\|P_Uu\|}{\|u\|}\right] = \frac{a\beta}{\|u\|}\big(\|u\| + \|P_Uu\|\big)\\
&= \frac{a\beta}{\|u\|}\cdot\frac{\big(\|u\| + \|P_Uu\|\big)\big(\|u\| - \|P_Uu\|\big)}{\|u\| - \|P_Uu\|}\\
&= \frac{\beta}{\|u\|^2}\cdot\frac{1}{\left[1 - \frac{\|P_Uu\|}{\|u\|}\right]} = \frac{\beta}{\|u\|^2(1 - \cos\theta(u,\operatorname{ran}U))}.
\end{align*}
Therefore, with $\tau = 1 - \cos\theta(u,\operatorname{ran} U)$,
\begin{align*}
    \lambda_{\min}(A^\top A) &= \|(A^\top A)^{-1}\|^{-1} \ge \tau\|u\|^2\cdot\min\left\{1,\frac{1}{\|u\|^2b}\right\} \\
    &= \tau\cdot\min\left\{\|u\|^2,\lambda_{\min}(U^\top U)\right\},
\end{align*}
which is the claim as $\lambda_{\min}(U^\top U) = \lambda_{\min}(UU^\top) = \lambda_{\min}(Q)$.
\end{proof}

\begin{proof}[Proof of Theorem~\ref{t:main}]The proof is carried out in two steps. In the first step, we shall derive a lower bound for $\sigma_\mathrm{min}(V)$ in terms of $\sigma_\mathrm{min}(U_m)$, namely
\begin{equation}\label{e:sigma_V-test}
    \sigma_\mathrm{min}^2(V) \ge (1 - \cos\theta(\one_{m+1},\operatorname{ran} Q))\cdot\min\big\{m+1,\,\sigma_\mathrm{min}^2(U_m)\big\},
\end{equation}
where $Q = U_{m+1}^\top U_{m+1}$.
If $V$ is not invertible, then $\one_{m+1}\in\operatorname{ran}U_{m+1}^\top = \operatorname{ran}Q$, and \cref{e:sigma_V-test} holds trivially. Assume that $V$ is invertible, and note that $V^\top V = P + Q$, where $P = \one_{m+1}\one_{m+1}^\top$. We apply Proposition~\ref{t:kaur_ext} and obtain
\begin{equation*}
\sigma_\mathrm{min}^2(V) = \lambda_\mathrm{min}(P+Q) \geq (1 - \cos\theta(\one_{m+1},\operatorname{ran}Q))\cdot\min\{\|\one_{m+1}\|^2,\,\lambda_{\min}(Q)\}.
\end{equation*}
Since $U_m$ is assumed to be invertible, we have
\begin{equation*}
        \lambda_{\min}(Q) = \lambda_{\min}(U_{m+1}U_{m+1}^\top) = \lambda_{\min}(U_mU_m^\top + u_0u_0^\top)\,\ge\,\lambda_{\min}(U_mU_m^\top) = \sigma_\mathrm{min}^2(U_m).
\end{equation*}
Hence, and as $\|\one_{m+1}\|^2 = m+1$, we arrive at \cref{e:sigma_V-test}. For completing the reduction from $\sigma_\mathrm{min}(V)$ to $\sigma_\mathrm{min}(U_m)$, the computation of the value $\cos\theta(\one_{m+1},\operatorname{ran}Q)$ remains. By definition of the cosine, we obtain
\[
\cos^2\theta(\one_{m+1},\operatorname{ran}Q) = \frac{\|P_{\operatorname{ran}Q}\one_{m+1}\|^2}{\|\one_{m+1}\|^2} = \frac{m+1 - \|P_{\ker Q}\one_{m+1}\|^2}{m+1}.
\]
Since $\ker Q = \operatorname{span}\{[\begin{smallmatrix} 1\\-v\end{smallmatrix}]\}$ with $v = U_m^{-1} u_0$, we have \[P_{\ker Q}x = \frac 1{1+\|v\|^2}(x^\top\left[\begin{smallmatrix}1\\-v\end{smallmatrix}\right])\left[\begin{smallmatrix}1\\-v\end{smallmatrix}\right],\] so that
\[
\cos^2\theta(\one_{m+1},\operatorname{ran}Q) = \frac{m+1 - \frac{(1 - \one_m^\top v)^2}{1+\|v\|^2}}{m+1} = (1 - \Theta(v))^2.
\]
Hence, $1 - \cos\theta(\one_{m+1},\operatorname{ran}Q) = \Theta(U_m^{-1}u_0)$.

In the second step, we estimate $\sigma_\mathrm{min}^2(U_m) = \lambda_{\min}(U_m^\top U_m)$ from below. For this, recall the enumeration $\{i_1,\ldots,i_m\}$ from Theorem \ref{t:main}. It is clear that a permutation of the columns of $U_m$ leaves the singular value invariant. Hence, we may assume that $i_s = s$ and $I_s = \{s+1,\ldots,m\}$ for each $s\in [1:m-1]$. We write $U_m = \begin{bmatrix}u_1 & U_{2:m}\end{bmatrix}$ and set $P = u_1u_1^\top$ as well as $Q = U_{2:m}U_{2:m}^\top$. By Proposition~\ref{t:kaur_ext},
\begin{equation*}
        \begin{split}
        \lambda_{\min}(U_mU_m^\top) &= \lambda_{\min}(P+Q)\\
        &\geq (1 - \cos\theta(u_1,\operatorname{ran}U_{2:m}))\cdot\min\big\{\|u_1\|^2,\,\lambda_{\min}(U_{2:m}U_{2:m}^\top)\big\},
        \end{split}
\end{equation*}
but observe that $\lambda_{\min}(U_{2:m}U_{2:m}^\top) = \min\{\|U_{2:m}x\|^2 : \|x\|=1\}\le \|U_{2:m}e_1\|^2 = \|u_2\|^2\le\|u_1\|^2$, where $e_1$ is the first canonical basis vector in $\R^{m-1}$. Hence, we obtain $\lambda_{\min}(U_mU_m^\top)\,\ge\,(1 - \cos\theta(u_1,S_{I_1}))\cdot\lambda_{\min}(U_{2:m}U_{2:m}^\top)$. 
Proceeding further in this way yields
\begin{equation*}
        \lambda_{\min}(U_mU_m^\top) \geq \prod_{s=1}^{m-1}(1 - \cos\theta(u_s,S_{I_s})) \cdot \underbrace{\lambda_{\min}(U_{m:m}U_{m:m}^\top)}_{= \|u_m\|^2}.
\end{equation*}
This completes the proof of Theorem \ref{t:main}.
\end{proof}

\end{document}